\documentclass[10pt,reqno]{amsart}

\usepackage{hyperref}
\usepackage{latexsym,amsmath, bm}
\usepackage{enumerate}
\usepackage{wrapfig, caption}
\usepackage{amsfonts}
\usepackage{amssymb}
\usepackage{geometry}
\usepackage{latexsym}
\usepackage{fixmath}
\usepackage{faktor}
\usepackage{mathtools}
\usepackage[titletoc]{appendix}
\usepackage{longtable}

\usepackage[boxed, linesnumbered, noend, noline]{algorithm2e}
\SetKwInput{KwData}{Input}
\SetKwInput{KwResult}{Output}

\usepackage[T1]{fontenc}
\usepackage{fourier}
\usepackage{bbm}
\usepackage{color}
\usepackage{pgfplots}
\usepackage{tikz}
\usetikzlibrary{shapes,decorations,arrows,calc,arrows.meta,fit,positioning}
\usepackage{graphicx}
\usepackage{fullpage}
\usetikzlibrary{decorations.pathreplacing, matrix}

\renewcommand{\vec}[1]{\boldsymbol{#1}}

\renewcommand{\subset}{\subseteq}
\renewcommand{\ln}{\log}

\newcommand{\oneplus}{V_1^+}
\newcommand{\oneminusminus}{V_1^{--}}

\newcommand{\zerominus}{V_0^-}
\newcommand{\zeroplus}{V_0^+}

\newcommand\minf{m_{\mathrm{inf}}}
\newcommand\malg{m_{\mathrm{alg}}}

\newcommand\vX{\vec X}

\newcommand\vY{\vec Y}
\newcommand\vm{{\vec m}}

\newcommand\vA{\vec A}

\newcommand\G{\vec G}

\newcommand\vH{\vec H}

\newcommand\KL[2]{D_{\mathrm{KL}}\bc{{{#1}\|{#2}}}}
\newcommand\SIGMA{\vec\sigma}

\newcommand\pl[1]{\textcolor{black}{#1}}
\newcommand\mhk[1]{\textcolor{black}{#1}}

\newcommand\cA{\mathcal{A}}
\newcommand\cB{\mathcal{B}}

\newcommand\cG{\Gamma}
\newcommand\cE{\mathcal{E}}

\newcommand\cN{\mathcal{N}}

\newcommand\cS{\mathcal{S}}
\newcommand\cT{\mathcal{T}}

\newcommand\cM{\mathcal{M}}

\newcommand\cP{\mathcal{P}}

\newcommand\cV{\mathcal{V}}
\newcommand\cW{\mathcal{W}}

\def\cR{{\mathcal R}}

\def\cE{{\mathcal E}}

\newcommand\eul{\mathrm{e}}
\newcommand\eps{\varepsilon}

\newcommand\Erw{\mathbb{E}}
\newcommand{\vecone}{\vec{1}}

\newcommand{\Bin}{{\rm Bin}}
\newcommand{\Mult}{{\rm Mult}}

\newcommand{\bink}[2] {{\binom{#1}{#2}}}
\newcommand\bc[1]{\left({#1}\right)}
\newcommand\cbc[1]{\left\{{#1}\right\}}
\newcommand\bcfr[2]{\bc{\frac{#1}{#2}}}
\newcommand{\bck}[1]{\left\langle{#1}\right\rangle}
\newcommand\brk[1]{\left\lbrack{#1}\right\rbrack}
\newcommand\scal[2]{\bck{{#1},{#2}}}

\newcommand\abs[1]{\left|{#1}\right|}

\newcommand\RR{\mathbb{R}}

\newcommand{\Whp}{W.h.p.}
\newcommand{\whp}{w.h.p.}

\newcommand\pr{\mathbb{P}} 
\renewcommand\Pr{\pr} 
\newcommand\Lem{Lemma}
\newcommand\Prop{Proposition}
\newcommand\Thm{Theorem}

\newcommand\Cor{Corollary}
\newcommand\Sec{Section}

\newtheorem{definition}{Definition}[section]

\newtheorem{theorem}[definition]{Theorem}
\newtheorem{lemma}[definition]{Lemma}
\newtheorem{proposition}[definition]{Proposition}
\newtheorem{corollary}[definition]{Corollary}

\newcommand{\floor}[1]{\left\lfloor#1\right\rfloor}
\newcommand{\ceil}[1]{\left\lceil#1\right\rceil}

\newcommand{\ind}[1]{{\pmb 1}\{#1\}}

\def\pr{{\mathbb P}}

\newcommand{\remove}[1]{}

\newcommand{\be}{\begin{equation}}
	\newcommand{\bel}[1]{\begin{equation}\lab{#1}\ }
		\newcommand{\ee}{\end{equation}}
	\newcommand{\bea}{\begin{eqnarray}}
		\newcommand{\eea}{\end{eqnarray}}
	\newcommand{\bean}{\begin{eqnarray*}}
		\newcommand{\eean}{\end{eqnarray*}}

	\newcommand{\mone}{\vm_1}
	\newcommand{\mzero}{\vm_0}

\begin{document}

	\title{Information-theoretic and algorithmic thresholds for group testing}
	
	\thanks{Supported by DFG CO 646/3 and Stiftung Polytechnische Gesellschaft. An extended abstract of this work appeared in the 2019 ICALP proceedings. A revised version is to appear in IEEE Transactions on Information Theory (Copyright (c) 2017 IEEE DOI: 10.1109/TIT.2020.3023377)}

	\author{Amin Coja-Oghlan, Oliver Gebhard, Max Hahn-Klimroth, Philipp Loick}
	
	\address{Amin Coja-Oghlan, {\tt acoghlan@math.uni-frankfurt.de}, Goethe University, Mathematics Institute, 10 Robert Mayer St, Frankfurt 60325, Germany.}
	
	\address{Oliver Gebhard, {\tt gebhard@math.uni-frankfurt.de}, Goethe University, Mathematics Institute, 10 Robert Mayer St, Frankfurt 60325, Germany.}
	
	\address{Max Hahn-Klimroth, {\tt hahnklim@math.uni-frankfurt.de}, Goethe University, Mathematics Institute, 10 Robert Mayer St, Frankfurt 60325, Germany.}
	
	\address{Philipp Loick, {\tt loick@math.uni-frankfurt.de}, Goethe University, Mathematics Institute, 10 Robert Mayer St, Frankfurt 60325, Germany.}
	
	\begin{abstract}
	In the group testing problem we aim to identify a small number of infected individuals within a large population.
We avail ourselves to a procedure that can test a group of multiple individuals, with the test result coming out positive iff at least one individual in the group is infected.
With all tests conducted in parallel, what is the least number of tests required to identify the status of all individuals?
In a recent test design [Aldridge et al.\ 2016] the individuals are assigned to test groups randomly \pl{with replacement}, with every individual joining an \pl{almost} equal number of groups.
We pinpoint the sharp threshold for the number of tests required in this randomised design so that it is information-theoretically possible to infer the infection status of every individual.
Moreover, we analyse two efficient inference algorithms. 
These results settle conjectures from [Aldridge et al.\ 2014, Johnson et al.\ 2019].
	\end{abstract}

	\maketitle

\section{Introduction}\label{Sec_intro}

\subsection{Background and motivation}
The group testing problem goes back to the work of Dorfman from the 1940s~\cite{Dorfman_1943}.
Among a large population a few individuals are infected with a rare disease.
The objective is to identify the infected individuals effectively.
At our disposal we have a testing procedure capable of not merely testing one individual, but several.
The test result will be positive if \pl{at least one individual} in the test group is infected, and negative otherwise;
all tests are conducted in parallel.
We are at liberty to assign a single individual to several test groups.
The aim is to devise a test design that identifies the status of every single individual correctly while requiring as small a number of tests as possible.
A recently proposed test design allocates the individuals to tests randomly ~\cite{Aldridge_2017,Aldridge_2014, Aldridge_2016,Johnson_2019,Mezard2008}.
To be precise, given integers $n,m,\Delta>0$ we create a random bipartite multi-graph 
by choosing independently for each of the $n$ vertices $x_1,\ldots,x_n$ `at the top' $\Delta$ neighbours among the $m$ vertices $a_1,\ldots,a_m$ `at the bottom' uniformly at random with replacement.
The vertices $x_1,\ldots,x_n$ represent the individuals, the $a_1,\ldots,a_m$ represent the test groups and an individual joins a test group iff the corresponding vertices are adjacent (see Figure \ref{fig_FigureGTexample}).
The wisdom behind this construction is that the expansion properties of the random bipartite graph precipitate virtuous correlations, facilitating inference.
Given $n$ and (an estimate of) the number $k$ of infected individuals, what is the least $m$ for which, with a suitable choice of $\Delta$, the status of every individual can be inferred correctly from the test results with high probability?Like in many other inference problems the answer comes in two instalments.
First, we might ask for what $m$ it is {\em information-theoretically} possible to detect the infected individuals.
In other words, regardless of computational resources, do the test results contain enough information in principle to identify the infection status of every individual?
Second, for what $m$ does this problem admit {\em efficient algorithms}?
The first main result of this paper resolves the information-theoretic question completely.
Specifically, Aldridge, Johnson and Scarlett~\cite{Aldridge_2016} obtained a function $\minf=\minf(n,k)$ such that for any fixed $\eps>0$ 
the inference problem is information-theoretically infeasible if $m<(1-\eps)\minf$. They conjectured that this bound is tight, i.e., that for $m>(1+\eps)\minf(n,k)$ there is an (exponential) algorithm that correctly identifies the infected individuals with high probability.
We prove this conjecture.  Furthermore, concerning the algorithmic question, Johnson, Aldridge and Scarlett~\cite{Johnson_2019} obtained a function $\malg=\malg(n,k)$ that exceeds $\minf$ by a constant factor \pl{for small $k$} such that for $m>(1+\eps)\malg$ certain efficient algorithms successfully identify the infected individuals with high probability.
They conjectured that {\tt SCOMP}, their most sophisticated algorithm, actually succeeds for smaller values of $m$.
We refute this conjecture and show that {\tt SCOMP} \mhk{asymptotically} fails to outperform a much simpler algorithm called {\tt DD}.
A technical novelty of the present work is that we investigate the group testing problem from a new perspective.
While most prior contributions rely either on elementary calculations and/or information-theoretic arguments~\cite{Aldridge_2014, Aldridge_2016, Johnson_2019,Scarlett_2016}, here we bring to bear techniques from the theory of random constraint satisfaction problems~\cite{ANP,MM}.
\newpage

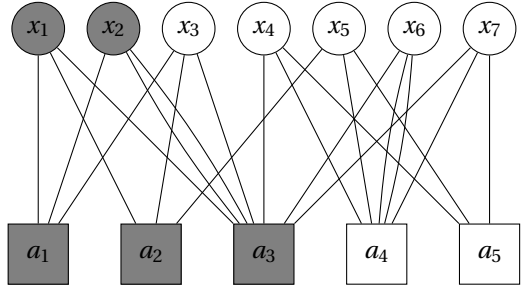
\begin{figure}
\captionsetup{margin=0.01cm}
\centering
\begin{tikzpicture}[scale=1]
\node[circle, draw, minimum width=0.7cm, fill=gray] (x0) at (0, 0) {$x_1$};
\node[circle, draw, minimum width=0.7cm, fill=gray] (x1) at (1,0) {$x_2$};
\node[circle, draw, minimum width=0.7cm] (x2) at (2,0) {$x_3$};
\node[circle, draw, minimum width=0.7cm] (x3) at (3,0) {$x_4$};
\node[circle, draw, minimum width=0.7cm] (x4) at (4,0) {$x_5$}; 
\node[circle, draw, minimum width=0.7cm] (x5) at (5,0) {$x_6$};
\node[circle, draw, minimum width=0.7cm] (x6) at (6,0) {$x_7$};

\node[rectangle, draw, minimum width=0.8cm, minimum height=0.8cm, fill=gray] (a1) at (0, -3) {$a_1$};
\node[rectangle, draw, minimum width=0.8cm, minimum height=0.8cm, fill=gray] (a2) at (1.5,-3) {$a_2$};
\node[rectangle, draw, minimum width=0.8cm, minimum height=0.8cm, fill=gray] (a3) at (3, -3) {$a_3$};
\node[rectangle, draw, minimum width=0.8cm, minimum height=0.8cm] (a4) at (4.5, -3) {$a_4$};
\node[rectangle, draw, minimum width=0.8cm, minimum height=0.8cm] (a5) at (6,-3) {$a_5$};

\path[draw] (x0) -- (a1);
\path[draw] (x0) -- (a2);
\path[draw] (x0) -- (a3);
\path[draw] (x1) -- (a1);
\path[-] (x1) edge [bend left=5] (a3);
\path[-] (x1) edge [bend right=5] (a3);
\path[draw] (x2) -- (a1);
\path[draw] (x2) -- (a3);
\path[draw] (x2) -- (a2);
\path[draw] (x3) -- (a3);
\path[draw] (x3) -- (a4);
\path[draw] (x3) -- (a5);
\path[draw] (x4) -- (a2);
\path[draw] (x4) -- (a5);
\path[draw] (x4) -- (a4);
\path[draw] (x5) -- (a3);
\path[-] (x5) edge [bend left=5] (a4);
\path[-] (x5) edge [bend right=5] (a4);
\path[draw] (x6) -- (a3);
\path[draw] (x6) -- (a5);
\path[draw] (x6) -- (a4);
\end{tikzpicture}
\caption{The graph illustrates a small example of a group testing instance, with the individuals $x_1,\ldots,x_7$\\ at the top and the tests $a_1,\ldots,a_5$ at the bottom.
Infected individuals and positive tests are coloured in grey.}
\label{fig_FigureGTexample}
\end{figure}

Indeed, group testing can be viewed naturally as a constraint satisfaction problem:
the tests provide the constraints and the task is to find all possible ways of assigning a status (`infected' or `not infected') to the $n$ individuals in a way consistent with the given test results.
Since the allocation of individuals to tests is random, this question is similar in nature to, e.g., the random $k$-SAT problem that asks for a Boolean assignment that satisfies a random collection of clauses~\cite{nae,yuval,KostaSAT,DSS3}.
It also puts the group testing problem in the same framework as the considerable body of recent work on other inference problems on random graphs such as the stochastic block model (e.g.,~\cite{AbbeSurvey,CKPZ,Decelle,CrisSurvey,Mossel,LF}) or decoding from pooled data \cite{alaoui_2019, alaoui_2019_2}.

We proceed to state the main results of the paper precisely, followed by a detailed discussion of the prior literature on group testing.
The proofs of the information-theoretic and algorithmic bounds follow in \ref{Sec_InfThUpper}, \Sec~\ref{Sec_InfThLower}, and \ref{Sec_alg}. \pl{The technical details can be found in the appendix.}

\subsection{The information-theoretic threshold}
Throughout the paper we labour under the assumptions commonly made in the context of group testing; we will revisit their merit in \Sec~\ref{Sec_related}.
Specifically, we assume that the number $k$ of infected individuals satisfies $k\sim n^\theta$ for a fixed $0<\theta<1$ \footnote{While we write that $k \sim n^{\theta}$ for the sake of brevity, our results immediately extend to the case $k \sim Cn^\theta$ for some constant $C$.}.
Moreover, let $\SIGMA\in\{0,1\}^{\{x_1,\ldots,x_n\}}$ be a vector of Hamming weight $k$ chosen uniformly at random.
The (one-)entries of $\SIGMA$ indicate which of the $n$ individuals are infected.
Moreover, let $\G=\G(n,m,\Delta)$ signify the aforementioned random bipartite graph \pl{with multi-edges}.
Then $\SIGMA$ induces a vector $\hat\SIGMA\in\{0,1\}^{\{a_1,\ldots,a_m\}}$ that indicates which of the $m$ tests come out positive.
To be precise, $\hat\SIGMA_i=1$ iff test $a_i$ is adjacent to an individual $x_j$ with $\SIGMA_{x_j}=1$.
For what $m$ is it possible to recover $\SIGMA$ from $\G,\hat\SIGMA$?
(Throughout the paper all logarithms are base $\eul$.)

\begin{theorem}\label{Thm_inf}
Suppose that $0<\theta<1$, $k \sim n^{\theta}$ and $\eps>0$ and let
\begin{align*}
\minf=\minf(n,\theta)&=\frac{k \log\bc{n/k}}{\min\cbc{1,\frac{1-\theta}\theta\log2}\log2}.
\end{align*}
\begin{enumerate}[(i)]
\item If $m>(1+\eps)\minf(n,\theta)$, then there exists an algorithm that given $\G,\hat\SIGMA$ outputs $\SIGMA$ with high probability.
\item If $m<(1-\eps)\minf(n,\theta)$, then there does not exist any algorithm that given $\G,\hat\SIGMA,k$ outputs $\SIGMA$ with a non-vanishing probability.
\end{enumerate}
\end{theorem}

Since for $\theta\leq\log(2)/(1+\log(2))$ the first part of \Thm~\ref{Thm_inf} readily follows from a folklore argument~\cite{DuHwang}, the interesting regime is $\theta>\log(2)/(1+\log(2))\approx 0.41$.
The negative part of \Thm~\ref{Thm_inf} strengthens a result from \cite{Aldridge_2016}, who showed that for $m<(1-\eps)\minf$ any inference algorithm has a strictly positive error probability.
By comparison, \Thm~\ref{Thm_inf} shows that any algorithm fails with {\em high} probability.

But the main contribution of \Thm~\ref{Thm_inf} is the first, positive statement.
While the problem was solved for $\theta<1/3$ for a different test design \cite{Scarlett_2016, Scarlett_2017} and the case $\theta>1/2$ is easy because a plain greedy algorithm succeeds~\cite{Johnson_2019}, the case $1/3<\theta<1/2$ proved more challenging. Only heuristic arguments \pl{predicting the result of \Thm~\ref{Thm_inf}} have been put forward for this regime so far \cite{Mezard2008}. Indeed, Aldridge et al.~\cite{Aldridge_2014} conjectured that in this case inferring $\SIGMA$ from $\G,\hat\SIGMA$ is equivalent to solving a hypergraph minimum vertex cover problem.
The proof of \Thm~\ref{Thm_inf} vindicates this conjecture.
Specifically, the vertex set of the hypergraph comprises all `potentially infected' individuals, i.e., those that do not appear in any negative test.
The hyperedges are the neighbourhoods $\partial a_i$ of the positive tests $a_i$ in $\G$.
Exhaustive search solves this vertex cover problem in time $\exp(O(n^\theta\log n))$.
But how about efficient algorithms for general $\theta$?

\subsection{Efficient algorithms for group testing} \label{sec_sub_algo}
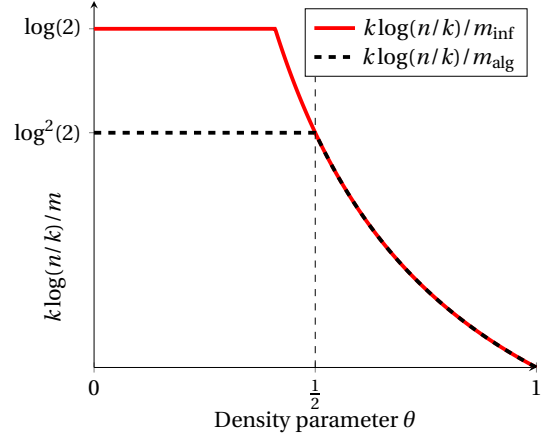
\begin{wrapfigure}{r}{0.5 \textwidth}
\captionsetup{margin=0.01cm}
\centering
\begin{tikzpicture}[scale=0.9]
\begin{axis}[
axis lines = left,
xlabel = Density parameter $\theta$,
ylabel = $k \log(n/k)/m$,
ylabel style={at={(0.1,0.3)}},
xtick={0, 0.5, 1},
xticklabels={$0$,$\frac{1}{2}$, $1$},
ytick={0.4804, 0.6931},
yticklabels={$\log^2(2)$,$\log(2)$},
ymin = 0,
ymax = 0.75,
xmin = 0,
xmax=1,
height=7cm
]
\addplot [
domain=0:1, 
samples=500, 
style={ultra thick},
color=red,
]
{min(ln(2), (1-x)/x * ln(2)^2)};
\addlegendentry{$k \log (n/k)/m_{\text{inf}}$}
\addplot [
domain=0:1, 
samples=500, 
color=black,
style={ultra thick},
dashed
]
{min(ln(2)^2, (1-x)/x * ln(2)^2)};
\addlegendentry{$k \log (n/k)/m_{\text{alg}}$}
\addplot+[
mark=none,
color=black,
dashed
]
coordinates
{(0.5,0) (0.5,0.7)};
\end{axis}
\end{tikzpicture}
\caption{The red line shows the information theoretic threshold $\minf$, the dashed black line signifies the bound $\malg$ which is achieved by the both the {\tt SCOMP} and the {\tt DD} algorithm.}
\label{fig_FigureThresholds}
\end{wrapfigure}
Several polynomial time group testing algorithms have been proposed.
A very simple greedy strategy called {\tt DD} (for `definitive defectives') first labels all individuals
that are members of negative test groups as uninfected.
Subsequently it checks for positive tests in which all individuals but one have been identified as uninfected in the first step.
Clearly, the single as yet unlabelled individual in such a test group must be infected.
Up to this point all decisions made by {\tt DD} are correct.
But in the final step {\tt DD} marks all as yet unclassified individuals as uninfected, possibly causing false negatives.
In fact, the output of {\tt DD} may be inconsistent with the test results as possibly some positive tests may fail to \pl{include} an individual classified as 'infected'. \pl{While an achievability result is known for the {\tt DD} algorithm, a corollary of the work in this paper is a matching converse.}

The more sophisticated {\tt SCOMP} algorithm is roughly equivalent to the well-known greedy algorithm for the hypergraph vertex cover problem applied to the hypergraph from the previous paragraph.
Specifically, in its first step {\tt SCOMP} proceeds just like {\tt DD}, classifying all individuals that occur in negative tests as uninfected.
Then {\tt SCOMP} identifies as infected all unmarked individuals that appear in at least one test whose other participants are already known to be uninfected.
Subsequently the algorithm keeps picking an individual that appears in the largest number of as yet `unexplained' (viz.\ uncovered) positive tests and marks that individual as infected, with ties broken randomly, until every positive test contains an individual classified as infected.
Clearly, {\tt SCOMP} may produce false positives as well as false negatives.
But at least the output is consistent with the test results.
\pl{Algorithm \ref{SCOMP_algorithm} summarises the procedure of {\tt SCOMP}.}

\IncMargin{1em}
\begin{algorithm}[ht]
\KwData{$\G$, $\hat\SIGMA$, $k$}
\KwResult{estimate of $\SIGMA$}
Classify all individuals in negative tests as healthy \& remove such individuals and tests from $\G$\;
Classify all individuals that appear in at least one positive test as the only yet unclassified individuals as infected \& remove such individuals and tests from $\G$\;
\While{there exists at least one test in $\G$}{Classify the individual appearing in the largest number of remaining tests as infected \& remove this individual and all adjacent tests from $\G$}
Classify all remaining individuals as healthy\;
\caption{Description of the {\tt SCOMP} algorithm} \label{SCOMP_algorithm}
\end{algorithm}
\DecMargin{1em}

Analysing {\tt SCOMP} has been prominently posed as an open problem in the group testing literature~\cite{Aldridge_2017b,Aldridge_2014, Johnson_2019}.
Indeed, Aldridge et al.~\cite{Aldridge_2014} opined that ``the complicated sequential nature of {\tt SCOMP} makes it difficult to analyse mathematically''.
On the positive side, \cite{Aldridge_2014} proved that {\tt SCOMP} succeeds in recovering $\SIGMA$ correctly given $(\G,\hat\SIGMA)$ if $m>(1+\eps)\malg(n,\theta)$ \whp \footnote{\pl{\Whp refers to a probability of $1-o(1)$ as $n \to \infty$.}}, where
\begin{align}\label{eqmalg}
\malg=\malg(n,\theta)&=\frac{k \log\bc{n/k}}{\min\cbc{1,\frac{1-\theta}\theta}\log^2 2}.
\end{align}
However, the algorithm succeeds for a trivial reason;
namely, for $m>(1+\eps)\malg$ even {\tt DD} suffices to recover $\SIGMA$ \whp\ 
Yet based on experimental evidence \cite{Aldridge_2014, Johnson_2019} conjectured that {\tt SCOMP} strictly outperforms {\tt DD}.
The following theorem refutes this conjecture.

\begin{theorem}\label{Thm_alg}
Suppose that $0<\theta<1$ and $\eps>0$.
If $m<(1-\eps)\malg(n,\theta)$, then given $\G,\hat\SIGMA$ {\whp\ both {\tt SCOMP} and {\tt DD} fail to output $\SIGMA$}.
\end{theorem}

For $\theta<1/2$ the information-theoretic bound provided by \Thm~\ref{Thm_inf} and the algorithmic bound $\malg$ supplied by \Thm~\ref{Thm_alg} remain a modest constant factor apart; see Figure~\ref{fig_FigureThresholds}.
\pl{Whether there exists an efficient algorithm for group testing that can close the gap to the information-theoretic bound has long been an open research question. A recent result by Coja-Oghlan et al. \cite{Coja_2019_SPIV} shows that such a polynomial-time algorithm indeed exists. The proposed algorithm which is inspired by the notion of spatial coupling from coding theory is able to recover $\SIGMA$ whenever $m>(1+\eps) \minf$. Moreover, the authors prove that below the information-theoretic threshold from \Thm~\ref{Thm_inf} no non-adaptive algorithm can succeed under any test design (not only the random regular test design considered here) thereby establishing the presence of an adaptivity gap in the group testing problem. An exciting avenue for future research is to investigate the merits of the results and techniques of this paper and \cite{Coja_2019_SPIV, Hahn_2019} for the noisy variant of group testing.}

\subsection{Discussion and related work}\label{Sec_related}
Dorfman's original group testing scheme, intended to test the American army for syphilis, was {\em adaptive}.
In a first round of tests each soldier would be allocated to precisely one test group.
If the test result came out negative, none of the soldiers in the group were infected.
In a second round the soldiers whose group was tested positively would be tested individually.
Of course, Dorfman's scheme was not information-theoretically optimal.
\pl{A first-order} optimal adaptive scheme that involves several test stages, with the tests conducted in the present stage governed by the results from the previous stages, is known~\cite{Alleman_2013, DuHwang}.
In the adaptive scenario the information-theoretic threshold works out to be
\begin{align*}
\minf^{\mathrm{adapt}}(n,\theta)=\frac{k \log\bc{n/k}}{\log 2}.
\end{align*}
The lower bound, i.e., that no adaptive design gets by with $(1-\eps)\minf^{\mathrm{adapt}}(n,\theta)$ tests, follows from a very simple information-theoretic consideration.
Namely, with a total of $m$ tests at our disposal there are merely $2^m$ possible test outcomes, and we need this number to exceed the count $\bink nk$ of possible vectors $\SIGMA$, i.e.,~\cite{Aldridge_2019}.

More recently there has been a great deal of interest in non-adaptive group testing, where the infection status of each individual is to be determined after just one round of tests
\cite{Aldridge_2019, Chen_2008, Emad_2014, Mezard2008}.
This is the version of the problem that we deal with in the present paper.
An important advantage of the non-adaptive scenario is that tests, which may be time-consuming, can be conducted in parallel.
Indeed, some of today's most popular applications of group testing are non-adaptive such as DNA screening \cite{Chen_2008, Kwang_2006, Ngo_2000} or protein interaction experiments \cite{Mourad_2013, Thierry_2006} in computational molecular biology.
The randomised test design that we deal with here is the best currently known non-adaptive design (in terms of the number of tests required).

The most interesting regime for the group testing problem is when the number $k$ of infected individuals scales as a power $n^\theta$ of the entire population.
Mathematically this is because in the linear regime $k=\Omega(n)$ the optimal strategy is to perform $n$ individual tests~\cite{Aldridge_2018} \pl{in order to achieve a vanishing error probability}. \pl{Similarly, the case of constant $k$ has been solved for some time \cite{Sebo_1985}}.
Thus, for $k$ linear in $n$ \pl{and $k$ constant} \mhk{the theory is already well established.}
But the sublinear case is also of practical relevance, as witnessed by Heap's law in epidemiology~\cite{Benz_2008} or biological applications \cite{Emad_2014}.

Apart from the randomised test design $\G$ where each individual chooses precisely $\Delta$ tests (with replacement), the so-called Bernoulli design assigns each individual to every test with a certain probability independently.
A considerable amount of attention has been devoted to this model, and its information-theoretic threshold as well as the thresholds for various algorithms have been determined~\cite{Aldridge_2017b, Aldridge_2017, Aldridge_2014, Scarlett_2016}.
However, the Bernoulli test design, while easier to analyse, \pl{for $\theta>1/3$} is provably inferior to the test design $\G$ that we study here.
This is because in the Bernoulli design there are likely quite a few individuals that participate in far fewer tests than expected due to degree fluctuations. \pl{We note that our proofs can easily be adapted to reprove the known results for the Bernoulli design. In fact, many technical parts of the proofs become significantly easier and shorter, since we can assume independence between tests, whereas for the constant-column design under consideration here gives rise to subtle dependencies between the tests. A significant portion of the tests is devoted to getting a handle o these dependencies.}

\subsection{Notation}

Throughout the paper $\G=\G(n,m,\Delta)$ denotes the random bipartite graph that describes which individuals take part in which test groups, the vector $\SIGMA\in\{0,1\}^{\cbc{x_1,\ldots,x_n}}$ encodes which individuals are infected, and $\hat\SIGMA\in\{0,1\}^{\cbc{a_1,\ldots,a_m}}$ indicates the test results.
Clearly, $\G$ is independent of $\SIGMA$.
Moreover, $k\sim n^\theta$ signifies the number of infected individuals.
Additionally, we write 
$$V=V_n=\cbc{x_1,\ldots,x_n}, \qquad  V_0=\cbc{x_i\in V:\SIGMA_{x_i}=0} \qquad \text{and} \qquad V_1=V\setminus V_0$$ for the set of all individuals, the set of uninfected and infected individuals, respectively.
For an individual $x\in V$ we write $\partial x$ for the multi-set of tests $a_i$ adjacent to $x$ \pl{with $\abs{\partial x} = \Delta$}.
Analogously, for a test $a_i$ we denote by $\partial a_i$ the multi-set of individuals that take part in the test \pl{and $\Gamma_i = \abs{\partial a_i}$}. \mhk{These are multi-sets since individuals are assigned to tests uniformly at random with replacement and therefore $\G$ features multi-edges \whp.}
\pl{Let $\cG$ be the vector $(\Gamma_i)_{i \in [m]}$.}
Furthermore, all asymptotic notation refers to the limit $n\to\infty$.
Thus, $o(1)$ denotes a term that vanishes in the limit of large $n$, while $\omega(1)$ stands for a function that diverges to $\infty$ as $n\to\infty$.
We also let $c,d>0$ denote reals such that
\begin{align*}
m &= c k \log(n/k) & 
\Delta &= d \log (n/k). 
\end{align*}
Later, we will prove that $c,d=\Theta(1)$ as $n\to\infty$ is optimal for inference. Finally, let $\Gamma_{\min}=\min_{i\in[m]}\Gamma_i$, $\Gamma_{\max}=\max_{i\in[m]}\Gamma_i$.
\pl{The following sections will outline the proofs of the information-theoretic bounds and the analysis of the {\tt SCOMP} algorithm and feature the important proofs. The technical details are left to the appendix}


\section{Getting started} \label{Sec_Groundwork}

The very first item on the agenda is to get a handle on the posterior distribution of $\SIGMA$ given $\G$ and $\hat\SIGMA$.
To this end, let $S_k(\G,\hat\SIGMA)$ be the set of all vectors $\sigma\in\{0,1\}^V$ of Hamming weight $k$ such that
\begin{align*}
\hat\SIGMA_{a_i}&=\vecone\cbc{\exists x\in\partial a_i:\sigma_x=1}&&\mbox{ for all $i\in[m]$}.
\end{align*}
In words, $S_k(\G,\hat\SIGMA)$ contains the set of all vectors $\sigma$ with $k$ ones that label the individuals infected/uninfected in a way consistent with the test results, \pl{i.e. that are "satisfying sets" \cite{Aldridge_2014, Aldridge_2019}}.
Let $Z_k(\G,\hat\SIGMA)=|S_k(\G,\hat\SIGMA)|$.
The following proposition shows that the posterior of $\SIGMA$ given $\G,\hat\SIGMA$ is uniform on 
$S_k(\G,\hat\SIGMA)$.

\begin{proposition}[\cite{Aldridge_2017}]\label{Prop_Nishi}
For all $\tau\in\cbc{0,1}^{\cbc{x_1,\ldots,x_n}}$ we have
$\displaystyle\pr\brk{\SIGMA=\tau\mid\G,\hat\SIGMA}=\frac{\vecone\cbc{\tau\in S_k(\G,\hat\SIGMA)}}{Z_k(\G,\hat\SIGMA)}.$
\end{proposition}

Adopting the jargon of the recent literature on inference problems on random graphs, we refer to \Prop~\ref{Prop_Nishi} as the {\em Nishimori identity}~\cite{CKPZ,LF}.
The proposition shows that apart from the actual test results, there is no further `hidden information' about $\SIGMA$ encoded in $\G,\hat\SIGMA$.
In particular, the information-theoretically optimal inference algorithm just outputs a uniform sample from $S_k(\G,\hat\SIGMA)$.
In effect, we obtain the following.

\begin{corollary}\label{Cor_Nishi}
\begin{enumerate}
\item If $Z_k(\G,\hat\SIGMA)=\omega(1)$ \whp, then for any algorithm $\cA$ we have 
	$$\pr\brk{\cA(\G,\hat\SIGMA,k)=\SIGMA}=o(1).$$
\item If $Z_k(\G,\hat\SIGMA)=1$ \whp, then there is an algorithm $\cA$ such that
	$$\pr\brk{\cA(\G,\hat\SIGMA,k)=\SIGMA}=1-o(1).$$
\end{enumerate}
\end{corollary}

\noindent
Both the positive and the negative part of \Cor~\ref{Cor_Nishi} assume that the precise number $k$ of infected individuals is known to the algorithm.
This assumption makes the negative part stronger, but weakens the positive part.
Yet we will see in due course how in the positive scenario the assumption that $k$ be known can be removed.

For the information-theoretic bound, the proof hinges on analysing the number of individuals that can be flipped without affecting the test results. We encounter two kinds of such individuals. The first kind consists of healthy individuals that only appear in positive tests and which we will denote by $V_0^+$. In symbols,
\begin{align}\label{eqV0+}
V_0^+&=\cbc{x_i\in V_0:\forall a\in\partial x_i\exists y\in\partial a:\SIGMA_y=1}.
\end{align}
Similarly, let $V_1^+$ be the set of all infected individuals $x_i$ such that every test in which $x_i$ occurs features another infected individual; in symbols,
\begin{align*}
V_1^+&=\cbc{x_i\in V_1:\forall a\in\partial x_i\exists y\in\partial a\setminus\cbc{x_i}:\SIGMA_y=1}.
\end{align*}
We think of the individuals in $V_0^+$ as the `potential false positives'.
Indeed, if for any $x_i\in V_0^+$ we obtain $\SIGMA'$ from $\SIGMA$ by setting $x_i$ to one, then $\SIGMA'$ will render the same test results as $\SIGMA$.
Similarly, the individuals in $V_1^+$ are potential false negatives.
\pl{For completeness, we also define $V_0^-$ and $V_1^-$ as
\begin{align} \label{eqV0-V1-}
V_0^- = V_0 \setminus V_0^+ \qquad \text{and} \qquad V_1^- = V_1 \setminus V_1^+    
\end{align}}
In the following, let us get a handle on the size of sets $V_0^+$ and $V_1^+$.
Specifically, we prove the following five statements.

\begin{proposition}\label{Lemma_V_all}
Let $c, d=\Theta(1)$. Then, the following statements hold \whp
\begin{enumerate}[(1)]
	\item $\abs\zeroplus=(1+n^{-\Omega(1)})n\bc{1-\exp(-d/c)}^\Delta.$ \label{Lemma_V0-}
	\item If $k(1-\exp(-d/c))^\Delta\geq n^{\Omega(1)}$, then $\abs\oneplus=n^{\Omega(1)}.$ \label{Lemma_V1plus}
	\item If $k(1-\exp(-d/c))^\Delta = o(1)$, then $\abs\oneplus=o(1).$ \label{Lemma_V1plus_reverse}
	\item If $c<\frac{\theta}{1-\theta} \frac{1}{\log^2 2}$, then $\abs\oneplus{},\abs\zeroplus{} = n^{\Omega(1)}. $  \label{Lemma_V++}
	\item If $c>\frac{\theta}{1-\theta} \frac{1}{\log^2 2}$, then $\abs\oneplus{} = o(1).$ \label{Lemma_V+_reverse}
\end{enumerate}

\end{proposition}

The proof of \Prop~\ref{Lemma_V_all}, while not fundamentally difficult, requires a bit of care because we are dealing with a random bipartite multi-graph whose (test-)degrees scale as a power of $n$.
In effect, the diameter of the bipartite graph is quite small and the neighbourhoods of different tests may have a sizeable intersection.
The technical workout follows in \Sec~\ref{subsec_start_prop}.
In the next step, let us get a handle on the size of the test degrees.

\begin{lemma}\label{Lemma_GammaMinMax}
	With probability at least $1-o(n^{-2})$ we have
	$$\Delta n/m-\sqrt{\Delta n/m}\log n\leq\Gamma_{\min}\leq\Gamma_{\max}\leq\Delta n/m+\sqrt{\Delta n/m}\log n.$$
\end{lemma}

The proof of this and the subsequent elementary lemmas are included in \Sec~\ref{Sec_sub_proof}. Next, we calculate the number of positive and negative tests. Let $\mone$ be the number of positive tests and let $\mzero$ be the number of negative tests.
Clearly $\mzero+\mone=m$.

\begin{lemma}\label{Lemma_m0}
	With probability at least $1-o(n^{-2})$ we have $$\mzero=\exp(-d/c)m+O(\sqrt{m}\log^2 n).$$
\end{lemma}

Finally, we justify that setting $c,d=\Theta(1)$ as $n \to \infty$ is optimal for inference. The fact that $c=\Theta(1)$ immediately follows from the information-theoretic counting bound, i.e.,~\cite{Aldridge_2019}.
\begin{lemma}\label{Lemma_Delta_too_small_big}
\begin{enumerate}[(1)]
	\item	If $\Delta=o(\log(n/k))$ and $m=\Theta(k\log(n/k))$, then $Z_k(\G,\hat\SIGMA)=\omega(1)$ \whp{}
	\item  If $\Delta=\omega(\log(n/k))$ and $m=\pl{\Theta}(k\log(n/k))$, then $Z_k(\G,\hat\SIGMA)=\omega(1)$ \whp{}
\end{enumerate}
\end{lemma}



\section{The information-theoretic upper bound}\label{Sec_InfThUpper}

We proceed to discuss the proof of \Thm~\ref{Thm_inf}.
The proof of the first, positive statement and of the second, negative statement hinge on two separate arguments.
We begin with the proof of the information-theoretic upper bound which is the principal achievement of the present work.
The proof rests upon techniques that have come to play an important role in the theory of random constraint satisfaction problems. 
Specifically, we need to show that $Z_k(\G,\hat\SIGMA)=1$ \whp, i.e., that $\SIGMA$ is the only assignment compatible with the test results \whp\
We establish this result by combining two separate arguments.
First, we use a moment calculation to show that \whp\ there are no other solutions that have a small `overlap' with $\SIGMA$.
Then we use an expansion argument to show that \whp\ there are no alternative solutions with a big overlap.
Both these arguments are variants of the arguments that have been used to study the solution space geometry of random constraint satisfaction problems such as random $k$-SAT or random $k$-XORSAT~\cite{geom,nae,DM}, as well as the freezing thresholds of random constraint satisfaction problems~\cite{Barriers,Molloy}.
Yet to our knowledge these methods have thus far not been applied to the group testing problem.
\pl{In this section we choose $\Delta=\lceil \frac{m}{k}\log 2\rceil$ which maximises the entropy of the test results.}
Formally, we define
\begin{align*}
Z_{k,\ell}(\G,\hat\SIGMA)&=\abs{\cbc{\sigma\in S_k(\G,\hat\SIGMA):
\scal{\SIGMA}\sigma=\ell}}
\end{align*}
as the number of assignments $\sigma\in S_k(\G,\hat\SIGMA)$ \pl{different from the true configuration $\SIGMA$} whose {\em overlap}
$$\scal{\SIGMA}\sigma=\sum_{i=1}^n\vecone\{\sigma_{x_i}=\SIGMA_{x_i}=1\}$$ 
with $\SIGMA$ is equal to $\ell$.
The following two propositions rule out assignments with a small and a big overlap, respectively.
In either case we choose $\Delta=\lceil \frac{m}{k}\log 2\rceil$ to take its optimal value.

\begin{proposition}\label{Prop_small_overlap}
Let $\eps>0$ and $0<\theta<1$ and assume that $m>(1+\eps)\minf(k,\theta)$.
\Whp\ we have $Z_{k,\ell}(\G,\hat\SIGMA)=0$ for all $\ell<(1-1/\log n)k$.
\end{proposition}

\begin{proof}
For $i\in[m]$ let $\Gamma_i$ be the degree of $a_i$ in $\G$, i.e., the number of edges incident with $a_i$; this number may exceed the number of different individuals that participate in test $a_i$ as $\G$ may feature multi-edges.
\mhk{Let $\cG$ be the $\sigma$-algebra generated by the random variables $(\Gamma_i)_{i\in[m]}$. Whenever we condition on $\Gamma$, we assume that the bounds from \Lem~\ref{Lemma_GammaMinMax} and \ref{Lemma_m0} hold.}
Given $\cG$ we can generate $\G$ from the well-known {\em pairing model}~\cite{JLR}.
Specifically, we create a set $\cbc{x_i}\times[\Delta]$ of $\Delta$ clones of each individual as well as sets $\cbc{a_i}\times[\Gamma_i]$ of clones of the tests.
Then we draw a perfect matching of the complete bipartite graph on the vertex sets $\bigcup_{i=1}^n\cbc{x_i}\times[\Delta]$, $\bigcup_{i=1}^m\cbc{a_i}\times[\Gamma_i]$ uniformly at random.
For each matching edge linking a clone of $x_i$ with a clone of $a_j$ we insert an $i$-$j$-edge.
The resulting bipartite random multi-graph has the same distribution as $\G$ given $\cG$. 
As an application of this observation we obtain for every integer $0\leq \ell<k$
\begin{align}\label{eqLemma_small_overlap}
\Erw[Z_{k,\ell}(\G,\hat\SIGMA)\mid\cG]&\leq O\bc{\bc{\Delta k}^{3/2}}\cdot\bink k\ell\bink{n-k}{k-\ell}
\prod_{i=1}^m \bc{1-2(1-k/n)^{\Gamma_i}+2\bc{1-2k/n+\ell/n}^{\Gamma_i}}
\end{align}

To see why \eqref{eqLemma_small_overlap} holds we use the linearity of expectation.
The product of the two binomial coefficients simply accounts for the number of assignments $\sigma$ that have overlap $\ell$ with $\SIGMA$.
Hence, with $\cS$ the event that one specific $\sigma\in\{0,1\}^{V}$ that has overlap $\ell$ with $\SIGMA$ belongs to $S_{k,\ell}(\G,\hat\SIGMA)$, we need to show that
\begin{align}\label{XeqLemma_small_overlap1}
\pr\brk{\cS\mid\cG}&\leq O\bc{\bc{\Delta k}^{3/2}} \prod_{i=1}^m1-2(1-k/n)^{\Gamma_i}+2\bc{1-2k/n+\ell/n}^{\Gamma_i}.
\end{align}
By symmetry we may assume that $\SIGMA_{x_i}=\vecone\{i\leq k\}$ and that
$\sigma_{x_i}=\vecone\{i\leq\ell\}+\vecone\{k<i\leq 2k-\ell\}$.

To establish \eqref{XeqLemma_small_overlap1} we harness the pairing model.
Namely, given $\cG$ we can think of each test $a_i$ as a bin of capacity $\Gamma_i$.
Moreover, we think of each clone $(x_i,h)$, $h\in[\Delta]$, of an individual as a ball.
The ball is labelled $(\SIGMA_{x_i},\sigma_{x_i})\in\{0,1\}^2$.
The random matching that creates $\G$ effectively tosses the $\Delta n$ balls randomly into the bins.
Hence, for $i\in[m]$ and for $j\in[\Gamma_i]$ let us write $\vA_{i,j}=(\vA_{i,j,1},\vA_{i,j,2})\in\{0,1\}^2$ for the label of the $j$th ball that ends up in bin number $i$.
Then we are left to calculate \pl{the probability that for every test $a_i$ either $\vA_{i,j,1}=\vA_{i,j,2}=0$ for every $j \in [\Gamma_i]$ or there is at least one pair $(j,k) \in [\Gamma_i]^2$ such that $\vA_{i,j,1}=\vA_{i,k,2}=1$}
\begin{align}\label{eqLemma_small_overlap2}
\pr\brk{\cS\mid\cG}&= \pr\brk{\forall i\in[m]:\max_{j\in[\Gamma_i]}\vA_{i,j,1}=\max_{j\in[\Gamma_i]}\vA_{i,j,2}\,\big|\,\cG},
\end{align}

To calculate this probability we borrow a trick from the analysis of the random $k$-SAT model~\cite{KostaSAT}. 
Namely, we consider a new set of $\{0,1\}^2$-valued random variables 
$\vA_{i,j}'=(\vA_{i,j,1}',\vA_{i,j,2}')$ such that $(\vA_{i,j}')_{i\in[m],j\in[\Gamma_i]}$ are mutually independent and such that
\begin{align*}
\pr\brk{\vA_{i,j}'=(1,1)}&=\ell/n,&
\pr\brk{\vA_{i,j}'=(0,1)}&=\pr\brk{\vA_{i,j}'=(1,0)}=(k-\ell)/n,\\
\pr\brk{\vA_{i,j}'=(0,0)}&=(n-2k+\ell)/n
\end{align*}
for all $i,j$. \mhk{Due to their independence, these multinomially distributed random variables are much easier to handle than $\vA_{i,j}$. It will turn out, that given a (not too unlikely) event, it suffices to analyse these independent variables instead of $\vA_{i,j}$. }
Now, let $\cT$ be the event that 
\begin{align} \label{eq_def_T}
\sum_{i=1}^m\sum_{j=1}^{\Gamma_i}\vecone\cbc{\vA_{i,j}'=(1,1)}&=\ell\Delta,
\qquad \sum_{i=1}^m\sum_{j=1}^{\Gamma_i}\vecone\cbc{\vA_{i,j}'=(0,0)}=(n-2k+\ell)\Delta,\\
\sum_{i=1}^m\sum_{j=1}^{\Gamma_i}\vecone\cbc{\vA_{i,j}'=(1,0)}&=
\sum_{i=1}^m\sum_{j=1}^{\Gamma_i}\vecone\cbc{\vA_{i,j}'=(0,1)}=(k-\ell)\Delta,
\end{align}
i..e, that all of the sums on the l.h.s.\ are {\em precisely} equal to their expected values.
Then $\vA'=(\vA_{i,j}')_{i,j}$ given $\cT$ is distributed precisely as $\vA=(\vA_{i,j})_{i,j}$.
Hence, \eqref{eqLemma_small_overlap2} yields
\begin{align}\label{eqLemma_small_overlap3}
\pr\brk{\cS\mid\cG}
&=\pr\brk{\forall i\in[m]:\max_{j\in[\Gamma_i]}\vA_{i,j,1}'=\max_{j\in[\Gamma_i]}\vA'_{i,j,2}\mid\cG,\cT}.
\end{align}

Thus, let $$\cA=\cbc{\forall i\in[m]:\max_{j\in[\Gamma_i]}\vA_{i,j,1}'=\max_{j\in[\Gamma_i]}\vA'_{i,j,2}}.$$
\pl{The grand idea is now to calculate the probability $\pr\brk{\cA\mid\cG}$. Subsequently, we employ Bayes' Theorem to derive a bound for the conditional probability $\pr\brk{\cA\mid \cT, \cG}$ for which we know by the above application of the balls-into-bins principle $$\pr\brk{\cS\mid\cG} = \pr\brk{\cA\mid \cT, \cG}.$$}
Because the $(\vA_{i,j}')_{i,j}$ are mutually independent, we can easily compute the unconditional probability $\pr\brk{\cA\mid\cG}$: by inclusion/exclusion,
\begin{align}\label{eqLemma_small_overlap4}
\pr\brk{\cA\mid\cG}&=
\prod_{i=1}^m \bc{1-2(1-k/n)^{\Gamma_i}+2(1-2k/n+\ell/n)^{\Gamma_i}}
\end{align}
{(the probability that $\max_{j\in[\Gamma_i]}\vA_{i,j,1}'=\max_{j\in[\Gamma_i]}\vA_{i,j,2}'=1$, i.e., both tests positive, equals one minus the probability that $\max_{j\in[\Gamma_i]}\vA_{i,j,1}'=0$ minus the probability that
$\max_{j\in[\Gamma_i]}\vA_{i,j,2}'=0$ plus the probability that $\max_{j\in[\Gamma_i]}\vA_{i,j,1}'=\max_{j\in[\Gamma_i]}\vA_{i,j,2}'=0$; then add the probability that $\max_{j\in[\Gamma_i]}\vA_{i,j,1}'=\max_{j\in[\Gamma_i]}\vA_{i,j,2}'=0$, i.e., both tests negative).}

Finally, to deal with the conditioning we use Bayes' rule:
\begin{align}\label{eqLemma_small_overlap5}
\pr\brk{\cA\mid\cT,\cG}&=\frac{{\pr\brk{\cA\mid\cG}\pr\brk{\cT\mid\cA,\cG}}}{{\pr\brk{\cT\mid\cG}}}.
\end{align}
Since the $(\vA_{i,j}')_{i,j}$ are independent, Stirling's formula yields
\begin{align*}
    \pr\brk{\cT\mid\cG}=\Omega\bc{(\Delta k)^{-3/2}}.  
\end{align*}
\mhk{A short justification can be found in \Sec~\ref{Sec_stirling}.}
Moreover, by definition we have 
$\pr\brk{\cT\mid\cA,\cG}\leq 1$.
Hence, \eqref{XeqLemma_small_overlap1} follows from \eqref{eqLemma_small_overlap3}--\eqref{eqLemma_small_overlap5}.
To complete the proof of the proposition, we claim that
\begin{align}\label{eqProp_small_overlap_0}
\sum_{0\leq\ell\leq \ceil{(1-1/\log n)k}} O\bc{\bc{\Delta k}^{3/2}} \bink k\ell\bink{n-k}{k-\ell}
\prod_{i=1}^m \bc{1-2(1-k/n)^{\Gamma_i}+2\bc{1-2k/n+\ell/n}^{\Gamma_i}}&=o(1).
\end{align}
To prove Equation \eqref{eqProp_small_overlap_0}, let $\alpha=\ell/k$. Using \Lem~\ref{Lemma_GammaMinMax} \mhk{and recalling $m=ck\log(n/k)$ and $\Delta=d \log(n/k)$}, we find
\color{black}
\begin{align}
	\Erw&[Z_{k,l}(\G,\SIGMA)] \leq O\bc{\bc{\Delta k}^{3/2}} \binom{k}{(1-\alpha)k} \binom{n-k}{(1-\alpha)k} \prod_{i=1}^m \bc{1-2\bc{1-\frac kn}^{\Gamma_i} + 2\bc{1-\frac {2k}{n} + \frac{\alpha k}{n}}^{\Gamma_i}} \nonumber\\
	&\leq O\bc{\bc{\Delta k}^{3/2}} \bc{\frac{e}{(1-\alpha)} \frac{en}{(1-\alpha) k}}^{(1-\alpha) k} \bc{1-2\bc{1-\frac kn}^{\Gamma_{\max}} + 2 \bc{1-\frac{2k}{n}+ \frac{\alpha k}{n}}^{\Gamma_{\min}}}^m \nonumber\\
	&\leq O\bc{\bc{\Delta k}^{3/2}} \left(\frac{e}{(1-\alpha)} \frac{en}{(1-\alpha) k} \right)^{(1-\alpha) k} \bigg(1-2\bc{1-\frac kn}^{\frac{n\log 2}{k}\bc{1+n^{-\Omega(1)}}} \\&\qquad  
	\qquad\qquad\qquad\qquad\qquad\qquad\qquad\qquad\qquad+ 2\bc{1-\frac{2k}{n}+\frac{\alpha k}{n}}^{\frac{n\log 2}{k}\bc{1+n^{-\Omega(1)}}}\bigg)^m \nonumber\\
	&\leq O\bc{\bc{\Delta k}^{3/2}} \left(\frac{e}{(1-\alpha)} \frac{en}{(1-\alpha) k} \right)^{(1-\alpha) k} \bc{1 - \bc{1 - 2^{-(1 - \alpha)} } \exp\bc{n^{-\Omega(1)}} }^m \nonumber\\
	&= O\bc{\bc{\Delta k}^{3/2}} \left(\frac{e}{(1-\alpha)} \frac{en}{(1-\alpha) k}(k/n)^{c \log(2)+ n^{-\Omega(1)}}(1+o(1))\right)^{(1-\alpha) k}\nonumber\\
	&= O\bc{\bc{\Delta k}^{3/2}} \left(\frac{e^2 (k/n)^{c \log(2) -1 + n^{-\Omega(1)}}}{(1-\alpha)^2} \right)^{(1-\alpha) k}.\label{Proof_small_first}
\end{align}

\color{black}
By the definition of $m > (1+\eps) \minf$ and $\ell < \ceil{k(1 - \log^{-1} n)}$, we have
\begin{align} \label{Proof_small_second}
    c \log 2 = 1+\eps \qquad \text{and} \qquad (1-\alpha)^2 \geq 1/\log^2 n
\end{align}
\color{black}
Moreover, as $\ell < \ceil{k(1 - \log^{-1} n)}$ we have $(1 - \alpha) k = \omega(1)$. Thus \eqref{Proof_small_second} implies that \eqref{Proof_small_first} tends to zero with $n \to \infty$.
Therefore, the proposition follows from Equations \eqref{Proof_small_first}, \eqref{Proof_small_second} and Markov's inequality.
\color{black}

\end{proof}

The argument from \Prop~\ref{Prop_small_overlap} does not extend to large overlaps (close to $k$) because the expression on the r.h.s.\ of \eqref{eqLemma_small_overlap} gets too large.
In other words, merely computing the expected number of solutions with a given overlap does not do the trick.
This `lottery phenomenon' is ubiquitous in random constraint satisfaction problems: for big overlap values rare solution-rich instances drive up the expected number of solutions~\cite{nae,ANP}. Fortunately, we can find a remedy.

\begin{proposition}\label{Prop_big_overlap}
Let $\eps>0$ and $0<\theta<1$ and assume that $m>(1+\eps)\minf(k,\theta)$.
\Whp\ we have $Z_{k,\ell}(\G,\hat\SIGMA)=0$ for all $(1-1/\log n)k\leq\ell<k$.
\end{proposition}

In order to cope with this issue we take another leaf out of the random CSP literature~\cite{Barriers,Molloy}.
Namely, we show that the solution $\SIGMA$ is locally rigid.
That is, the expansion properties of the random bipartite graph $\G$ preclude the existence of other solutions that have a big overlap with $\SIGMA$.
The following lemma holds the key to this effect.

\begin{lemma}\label{Lemma_rigid}
For any $\eps>0$ there exists $\delta=\delta(\eps)>0$ such that for all $m>(1+\eps)\minf$ the following is true.
Let $\cR$ be the event that for every $x_i$ with $\SIGMA_{x_i}=1$ there are at least $\delta\Delta$ tests $a\in\partial x_i$ {such that $\partial a\setminus\{x_i\}\subset V_0$.}
Then $\pr\brk\cR=1-o(1)$.
\end{lemma}

\begin{proof}
	Let $(\vX_i)_{i\in[m]}$ be a sequence of independent $\Bin(\Gamma_i,k/n)$-variables as in \Sec~\ref{Sec_Groundwork}.
	Also let $W=\sum_{i=1}^m\vecone\cbc{\vY_i=1}$ as in \Sec~\ref{Sec_Groundwork}.
	Proceeding along the lines of the \mhk{proof of \Lem~\ref{Lemma_V_all} (see \eqref{eqLemma_V1plus_10} in \Sec~\ref{subsec_start_prop})}, we obtain
	\begin{align}\label{eqLemma_V1delta_1}
	\pr\brk{W=\bc{1+n^{-\Omega(1)}}k\Delta/2\mid\cG} &=1-o(n^{-7}).
	\end{align}
	Let $T$ be the number of \mhk{infected individuals which only show up less than $\delta \Delta$ of their tests as the only infected individual}, i.e. 
	$$ T = \abs{x\in V_1 : \sum_{a\in\partial x}\vecone\cbc{\partial a\setminus\cbc x\subset V_0}<\delta\Delta}.$$
	Moreover, let $\vH=\vH(N,K,n')$ be a hypergeometric random variable with parameters $N=k\Delta$ \mhk{(total eligible assignments for infected individuals)}, $K=W$ \mhk{(tests with only one infected individual)} and \mhk{$n'=\Delta$ (number of tests per individuals)}.
	Then \mhk{the union bound over $k$ infected individuals yields}
	\begin{align}\label{eqLemma_V1delta_2}
	\Erw\brk{T\mid\cG,W}&\leq k\pr\brk{\vH<\delta\Delta}.
	\end{align}
	Further, the Chernoff bound for the hypergeometric distribution \mhk{implies}
	\begin{align}\label{eqLemma_V1delta_3}
	\pr\brk{\vH<\delta\Delta}&\leq\exp(-\Delta\KL{\delta}{W/(k\Delta)})
	\end{align}
	\mhk{Recall $\Delta=d \log(n/k)$.} Since $\KL{\delta}{1/2+o(1)}=\delta\log\delta+(1-\delta)\log(1-\delta)+\log2+o(1)$ and $\delta\log\delta+(1-\delta)\log(1-\delta)\nearrow0$ as $\delta\to0$ and $c>\frac{\theta}{(1-\theta)\log^2 2}$, we can choose $\delta>0$ small enough so that
	\begin{align}\label{eqLemma_V1delta_4}
	\Delta(\delta\log\delta+(1-\delta)\log(1-\delta)+\log2+o(1))>\log k
    \end{align}
Finally, the assertion follows from \eqref{eqLemma_V1delta_1}--\eqref{eqLemma_V1delta_4}.
\end{proof}

Hence, \whp\ any infected individual appears in plenty of tests where all the other individuals are uninfected.
This property causes $\SIGMA$ to be locally rigid.
To see why, consider the repercussions of just changing the status of a single individual $x_i$ from infected to uninfected.
Because given $\cR$ the individual $x_i$ appears as the only infected individual in at least $\delta\Delta$ tests, in order to maintain the same tests results we will also need to flip at least one individual in each of these tests from `uninfected' to `infected'.
Since tests typically have relatively few individuals in common, the necessary number of flips from $0$ to $1$ will be $\Omega(\Delta)=\Omega(\log n)$.
But then in order to keep the total number of infected individuals constant $k$, we will need to perform another $\Omega(\Delta)$ flips from $1$ to $0$.
Yet given $\cR$ each of these `second generation' individuals that we flip from infected to uninfected is itself the only infected individual in many tests.
Thus, the single flip that we started from triggers a veritable avalanche of flips, which will stop only after the overlap has dropped significantly.
The next lemma formalises this intuition.
The lemma shows that while the unconditional expectation of $Z_{k,\ell}(\G,\hat\SIGMA)$ is `too big', the conditional expectation of $Z_{k,\ell}(\G,\hat\SIGMA)$ given $\cR$ \mhk{(as defined in \Lem~\ref{Lemma_rigid})} is much smaller.
Let $\mzero=\mzero(\G,\hat\SIGMA)$ be the total number of negative tests.

\begin{lemma}\label{Lemma_big_overlap}
Suppose that $(1-1/\log n) k\leq \ell< k$ and let 
$\Gamma_{\min} = \min_{i \in [m]} \Gamma_i$,
$\Gamma_{\max} = \max_{i \in [m]} \Gamma_i$.
Then
\begin{align}\label{eqLemma_big_overlap}
\Erw[Z_{k,\ell}(\G,\hat\SIGMA)\mid\cG,\cR,\mzero]&\leq O\bc{\bc{\Delta k}^{3/2}} \bink{k}{\ell}\bink{n-k}{k-\ell}
\bc{1-\bc{1-\frac{k-\ell}{n-k}}^{\Gamma_{\max}}}^{\delta \Delta (k-\ell)}\bc{\frac{n-2k+\ell}{n-k}}^{\bc{1+n^{-\Omega(1)}}\Gamma_{\min}\mzero}.
\end{align}
\end{lemma}

\mhk{The proof of \Lem~\ref{Lemma_big_overlap} is somehow subtle as we need to get a handle on the dependencies in $\G$ and is included in \Sec~\ref{sec_big_overlap}.}
\pl{To convey the intuition behind the expression in \Lem~\ref{Lemma_big_overlap}, the term $\bink{k}{\ell}\bink{n-k}{k-\ell}$ accounts for the number of assignments $\tau\in\cbc{0,1}^V$ of Hamming weight $k$ whose overlap with $\SIGMA$ is equal to $\ell$.
The terms thereafter capture the probability that such an assignment $\tau$ exhibits the same test results as the true configuration $\SIGMA$.
The first term provides a necessary condition for a positive test under $\SIGMA$ to stay positive under $\tau$. By \Lem~\ref{Lemma_rigid}, we know that every infected individual shows up in at least $\delta \Delta$ tests as the only infected individual. Now, there are $k-\ell$ infected under $\SIGMA$, but healthy under $\tau$. For any of these $\delta \Delta (k-\ell)$ tests, we need to have at least one individual that is healthy under $\SIGMA$, but infected under $\tau$ included in this test.
Next, we need to ensure that any negative test under $\SIGMA$ stay negative under $\tau$. To this end, every individual included in a negative test under $\SIGMA$ of which we have at least $\Gamma_{\min} \mzero$ must be healthy under $\tau$. The second term captures this probability. }

\begin{proof}[Proof of \Prop~\ref{Prop_big_overlap}]
In order to establish the proposition it suffices to show that there is $\eps' \leq (1 - 1/\log(n))k$ such that
\begin{align}
\sum_{\eps'\leq \ell \leq k} \Erw&[Z_{k,\ell}(\G,\hat\SIGMA) | \cG, \cR, \mzero] = o(1). \label{EqShowBigOverlap}
\end{align}
\color{black}
Starting from the expression in \Lem~\ref{Lemma_big_overlap}, setting $\alpha = \ell/k$ \mhk{and recalling $m=ck\log(n/k)$ and $\Delta=d \log(n/k)$}, we obtain
\begin{align}
\Erw&[Z_{k,\ell}(\G,\hat\SIGMA) | \cG, \cR, \mzero] \nonumber\\
&\leq O\bc{\bc{\Delta k}^{3/2}} \binom{k}{k-\ell} \binom{n-k}{k-\ell}
\bcfr{n-2k+\ell}{n-k}^{\bc{1+n^{-\Omega(1)}}\Gamma_{\min}\mzero} \bc{1-\bc{1-\frac{k-\ell}{n-k}}^{\Gamma_{\max}}}^{\delta\Delta(k-\ell)} \nonumber\\
&\leq O\bc{\bc{\Delta k}^{3/2}} \bc{\frac{e}{1-\alpha}}^{(1-\alpha)k} \bc{\frac{e(n-k)}{(1-\alpha)k}}^{(1-\alpha)k} \bc{1-\frac{(1-\alpha)k}{n-k}}^{\frac{mn \log 2}{2k} \bc{1+n^{-\Omega(1)}}} \bc{1 - 2^{-(1-\alpha) \bc{1+n^{-\Omega(1)}}}}^{\delta \Delta (1-\alpha)k}  \\
& \leq O\bc{\bc{\Delta k}^{3/2}} \left( \frac{e^2 n}{(1-\alpha)^2 k}\right)^{(1 - \alpha)k} \exp \bc{ (1 - \alpha) k \frac{c \log 2}{2} \bc{1+n^{-\Omega(1)}} \log(k/n) } \nonumber\\
& \qquad \qquad \qquad \cdot \exp \left( -c \delta \log (2) \log \left( 1 - 2^{-(1-\alpha) \bc{1+n^{-\Omega(1)}}} \right) \log(k/n) (1 - \alpha) k \right) \nonumber\\
& \leq O\bc{\bc{\Delta k}^{3/2}} \bc{ \frac{e^2 n}{(1-\alpha)^2 k} \exp \bc{ \log(k/n) \bc{1+n^{-\Omega(1)}} \bc{ \frac{c \log 2}{2} - c \delta \log (2) \log \bc{ 1 - 2^{-(1-\alpha)\bc{1+n^{-\Omega(1)}}} }}}}^{(1-\alpha) k}. \label{Proof_big_first}
\end{align}
As long as $1-\alpha=o(1)$, we find
$$(k/n)^{-\log \left(1-2^{-(1-\alpha)}\right)}(1-\alpha)^{-2} \to 0 \qquad \qquad \text{as } n \to \infty.$$
Moreover, $(1-\alpha) k \geq 1$. Thus, the expression \eqref{Proof_big_first} is of order
\begin{align}
O\bc{\bc{\Delta k}^{3/2}} \bc{k/n}^{\omega(1)} = n^{-\omega(1)}. \label{Proof_big_second}
\end{align}
Since \eqref{Proof_big_second} holds for any constant $c>0$ and any value of $\alpha$ s.t. $1-\alpha = o(1)$, it also holds for $\alpha \geq 1 - 1/\log n$. Consequently \eqref{EqShowBigOverlap} is established \whp
\color{black}
\end{proof}

\Prop s~\ref{Prop_small_overlap} and~\ref{Prop_big_overlap} readily imply that $Z_k(\G,\hat\SIGMA)=1$ \whp\ if $m>(1+\eps)\minf(k,\theta)$.
Hence, \Cor~\ref{Cor_Nishi} shows that there exists an inference algorithm that given $\G,\hat\SIGMA$ 
\emph{and $k$} outputs $\SIGMA$ \whp\ 
Up to now, the algorithm relies on exactly knowing the number of infected individuals $k$, which in practice could be rather difficult to learn. Fortunately, this assumption can be removed.
Namely, the following proposition shows that \whp\ there is no assignment $\sigma$ that is compatible with the test results and that has Hamming weight less than $k$.

\begin{proposition}\label{Prop_min}
Let $\eps>0$ and $0<\theta<1$ and assume that $m>(1+\eps)\minf(k,\theta)$.
\Whp\ we have $\sum_{k'<k}Z_{k'}(\G,\hat\SIGMA)=0$.
\end{proposition}

\begin{proof}
To get started, suppose that $0<\theta<1$ and $c<\log^{-2}2$. We claim that for any value of $d>0$, 
$\abs{\zeroplus}\geq k\log n$ \whp. 
Indeed, from \Prop~\ref{Lemma_V_all}(\ref{Lemma_V0-}), we know that
$$\abs\zeroplus=\bc{1+n^{-\Omega(1)}}n\bc{1-\exp(-d/c)}^\Delta. $$
\mhk{Recalling $\Delta=d \log(n/k)$,} the expression takes the minimum at $d=c \log 2$. It follows that
$$ \abs\zeroplus \geq \bc{1+n^{-\Omega(1)}}n (k/n)^{c \log^2 2}. $$
If $c=(1-\eps) \log^{-2}2$ for $\epsilon>0$, then
\begin{align}\label{0+_threshold}
   \abs\zeroplus \geq \bc{1+n^{-\Omega(1)}}n (k/n)^{1-\epsilon} = \bc{1+n^{-\Omega(1)}}k n^{(1-\theta)\epsilon} \geq k \log n \qquad \whp 
\end{align}
Now, the following two statements establish that if there does not exist a second satisfying set of Hamming weight $k$, there does also not exist a satisfying set with smaller Hamming weight \whp.

First, we claim that if $m > (1+\epsilon) \minf(k,\theta)$, \whp\ there does not exist a satisfying configuration with Hamming weight smaller than the correct configuration, where the set of infected individuals is not a subset of the true set of infected individuals. To see why, suppose there existed a satisfying configuration with a smaller Hamming weight, whose infected individuals are not a subset of the true infected individuals. By \eqref{0+_threshold}, we know that $\abs{\zeroplus} \gg k$ for $m < (1-\epsilon) \malg$ \whp\ Therefore, we could construct a satisfying configuration of identical Hamming weight as the true configuration by flipping individuals in $\zeroplus$ from healthy to infected. Observe that by the definition of $\zeroplus$, flipping individuals in $\zeroplus$ does not change the test result. Therefore, we would be left with a second satisfying configuration of identical Hamming weight as the true configuration, a contradiction to Propositions \ref{Prop_small_overlap} and \ref{Prop_big_overlap}.

Second, we argue that if $m > (1+\epsilon) \minf(k,\theta)$, \whp\ there does not exist a satisfying configuration with Hamming weight smaller than the correct configuration, where the set of infected individuals is a subset of the true set of infected individuals. Suppose there existed a satisfying configuration with a smaller Hamming weight, whose infected individuals are a subset of the true infected individuals. Then, the true configuration would need to contain individuals in $\oneplus$, which can be flipped from infected to healthy without affecting the test result. However, \Prop~\ref{Lemma_V_all}(\ref{Lemma_V+_reverse}) shows that for $m > (1+\epsilon)\minf$, $\oneplus=\emptyset$ \whp\
\end{proof}

As an immediate consequence of \Prop~\ref{Prop_min} we conclude that for $m>(1+\eps)\minf(k,\theta)$ the problem of inferring $\SIGMA$ boils down to a minimum vertex cover problem, as previously conjectured by 
Aldridge, Baldassini and Johnson~\cite{Aldridge_2014}.
Namely, let $\cP$ be the set of all positive tests, i.e., all tests $a_i$, $i\in[m]$, with $\hat\SIGMA_{a_i}=1$.
Moreover, let $V^+$ be the set of all variables $x_i\in V$ such that $\partial x_i\subset\cP$; 
in words, $x_i$ takes part in positive tests only.
We set up a hypergraph $\vH$ with vertex set $V^+$ and hyperedges $\partial a_i\cap V^+$, $a_i\in\cP$.
Clearly, the set of all individuals $x_i$ with $\SIGMA_{x_i}=1$ provides a valid vertex cover of $\vH$ (as any positive test must feature an infected individual).
Conversely, \Prop s~\ref{Prop_small_overlap} and~\ref{Prop_big_overlap} show that \whp\ this is the unique vertex cover of size $k$, and \Prop~\ref{Prop_min} shows that there is no strictly smaller vertex cover \whp\ 
Therefore, \whp\ we can infer $\SIGMA$ even without prior knowledge of $k$ by way of solving this minimum vertex cover instance.

\section{The information-theoretic lower bound}\label{Sec_InfThLower}

We proceed with the negative statement that \whp\ $\SIGMA$ cannot be inferred if $m<(1-\eps)\minf$.
In light of \Cor~\ref{Cor_Nishi} in order to prove the first part of \Thm~\ref{Thm_inf} we need to show that the number $Z_k(\G,\hat\SIGMA)$ of assignments consistent with the test results $\hat\SIGMA$ is unbounded \whp\ 
The proof of this fact is based on a very simple idea: we just identify a \pl{moderately large} number of individuals whose infection status could be flipped without affecting the test results.
The following lemma yields a bound on $m$ below which the number of such potential false positives ($\abs{V_0^+}$) and negatives ($\abs{V_1^+}$) abound.

\begin{proposition}\label{Lemma_V+}
Let $\eps>0$ and $0<\theta<1$ and assume that
\begin{align*}
m<\frac{(1-\eps) \theta}{(1-\theta) \log^2 2} n^{\theta} (1-\theta) \log n.
\end{align*}
Then for any choice of $\Delta$ we have $|V_0^+|,|V_1^+|=n^{\Omega(1)}$ \whp
\end{proposition}

\begin{proof}
Thanks to \Lem~\ref{Lemma_Delta_too_small_big} we may assume that $\Delta= d(\log(n/k))$, for a constant $d$ as this choice minimizes the number of individuals in $V_1^+$. Then \Prop~\ref{Lemma_V_all}(\ref{Lemma_V++}) guarantees that for every such constant as long as $c < \frac{\theta}{1-\theta}\frac{1}{\log^22}$, there are $n^{\Omega(1)}$ individuals in both $\oneplus$ and $\zeroplus$, which yields to \Prop~\ref{Lemma_V+}.
\end{proof}

As an immediate application we obtain the following information-theoretic lower bound.

\begin{corollary}\label{Cor_V+}
Let $\eps>0$ and $0<\theta<1$ and assume that
\begin{align}\label{eqCor_V+}
m<\frac{(1-\eps)\theta}{(1-\theta) \log^2 2} n^{\theta} (1-\theta) \log n.
\end{align}
Then $Z_k(\G,\hat\SIGMA)=\omega\bc1$ \whp
\end{corollary}
\begin{proof}
We need to exhibit alternative vectors $\SIGMA'\in\cbc{0,1}^V$ with Hamming weight $k$ that render the same test results as $\SIGMA$.
Thus, pick any $x_i\in V_0^+$ and any $x_j\in V_1^+$ and obtain $\SIGMA'$ from $\SIGMA$ by setting $\SIGMA'_{x_i}=1$ and $\SIGMA'_{x_j}=0$.
By construction, $\SIGMA'$ has Hamming weight $k$ and renders the same test results.
Hence, \Prop~\ref{Lemma_V+} shows that $Z_k(\G,\hat\SIGMA)\geq|V_0^+\times V_1^+|=\Omega(n^{2\theta})\gg1$ \whp
\end{proof}
The bound~\eqref{eqCor_V+} matches $\minf$ for $\theta \gtrapprox 0.41$.
A simpler, purely information-theoretic argument covers the remaining $\theta$.

\begin{proposition}\label{Lemma_classic}
Let $\eps>0$, $0<\theta<1$. If $m<\frac{1-\eps}{\log 2} n^{\theta} (1-\theta) \log n$, then
 $Z_k(\G,\hat\SIGMA)=\omega\bc1$ \whp
\end{proposition}

\begin{proof}
This Lemma follows from the classical information-theoretic lower bound for the group testing problem. Namely, $m$ tests allow for $2^m$ possible test results.
Hence, if $$m<\frac{(1-\eps)}{\log 2} n^{\theta} (1-\theta) \log n,$$ then the number of possible test results is far smaller than the number of vectors $\SIGMA\in\cbc{0,1}^V$ with Hamming weight $k$. Therefore, \whp\ there exists an unbounded number of vectors of Hamming weight $k$ that render the same test results as $\SIGMA$.
\end{proof}

We thus conclude that for all $0<\theta<1$, \whp\ $Z_k(\G,\hat\SIGMA)=\omega(1)$ if $m<(1-\eps)\minf$.
Therefore, the desired information-theoretic lower bound follows from \Cor~\ref{Cor_Nishi}.

\section{The {\tt SCOMP} algorithm} \label{Sec_alg}

\color{black}
For $\theta \geq 1/2$ we have $\malg=\minf$ and thus \Thm~\ref{Thm_inf} implies that {\tt SCOMP} \pl{as described in \Sec~\ref{sec_sub_algo}} \whp\ fails to infer $\SIGMA$ for $m<(1-\eps)\malg$.
Therefore, we are left to establish \Thm~\ref{Thm_alg} for $\theta<1/2$, in which case
\begin{align}\label{eqmalg_SCOMP}
\malg= \frac{k \log(n/k)}{\log^2 2}.
\end{align}

The proof of \Thm~\ref{Thm_alg} for $\theta<1/2$ hinges on two propositions.
First we show that below $\malg$, the set $V_1^{--}$ of infected individuals that the second step of {\tt SCOMP} identifies correctly is empty.
Formally, with $V_0^-$ from~\eqref{eqV0-V1-}, let 
\begin{align*}
V_1^{--}&=\cbc{x\in V_1:\exists a\in\partial x:\partial a\setminus\cbc x\subset V_0^-}.
\end{align*}

\begin{proposition}\label{Lemma_--}
Suppose that $0<\theta<1/2$ and $\eps>0$.
If $m<(1-\eps)\malg$, then for all $\Delta>0$ we have $V^{--}_1(\G,\hat\SIGMA^*)=\emptyset$ \whp
\end{proposition}

\mhk{The proofs of \Prop s~\ref{Lemma_--} and~\ref{Lemma_+} are based on moment calculations that turn out to be mildly subtle due to the potentially very large degrees of the underlying graph $\G$. The technical workout in included in \Sec~\ref{Proof_of_Lemma_--} and \ref{Proof_of_Lemma_+}.}

With the second step of {\tt SCOMP} failing to `explain' (viz.\ cover) any positive tests, the greedy vertex cover algorithm takes over.
This algorithm is applied to the hypergraph whose vertices are the as yet unclassified individuals and whose edges are the neighbourhoods of the positive tests.
Our second lemma shows that the set $V^{+,\Delta}$ of potententially false positive individuals $x\in\zeroplus$
that participate in the maximum number $\Delta$ of different tests is far greater than the actual number $k$ of infected individuals.
Formally, let
\begin{align*}
V_0^{+,\Delta}&=\cbc{x \in V_0^+: \abs{\partial x}=\Delta}.
\end{align*}

\begin{proposition}\label{Lemma_+}
Suppose that $0<\theta<1/2$ and $\eps>0$.
If $m<(1-\eps)\malg$, then for \mhk{$\Delta=d \log \bc{n/k}$ for all constant $d$} we have $\abs{V_0^{+,\Delta}}\geq k\log n$ \whp
\end{proposition}

We complete the proof of \Thm~\ref{Thm_alg} as follows.

\begin{proof}[Proof of \Thm~\ref{Thm_alg}]
The first step of {\tt SCOMP} (correctly) marks all individuals that appear in negative tests as healthy.
Moreover, \Prop~\ref{Lemma_--} implies that the second step of {\tt SCOMP} is void \whp, because there is no single infected individual that appears in {a} test whose other individuals have already been identified as healthy by the first step.
Consequently, {\tt SCOMP} simply applies the greedy vertex cover algorithm.
Now, thanks to \Prop~\ref{Lemma_+} it suffices to prove that {\tt SCOMP} will fail \whp\ if $\abs{V_0^{+,\Delta}}=\omega\bc k$.
Because they belong to positive tests only, all the individuals of $V_0^{+,\Delta}$ are present in the vertex cover instance that {\tt SCOMP} attempts to solve.
Moreover, in the hypergraph no vertex has degree greater than $\Delta$, because the degrees of $x_1,\ldots,x_n$ in $\G$ are equal to $\Delta$.
(Some of the hypergraph degrees may be strictly smaller than $\Delta$ because $\G$ is a multi-graph.)
Therefore, 
since $|V_0^{+,\Delta}|\geq k\log n$ while the actual set of infected individuals only has size $k$, \whp\ the individual classified as infected by the very first step of the greedy set cover algorithm belongs to $V^+_0$.
Hence, this individual is not actually infected, i.e., {\tt SCOMP} errs \whp\ 
\end{proof}

Since the success probability of the {\tt SCOMP} algorithm is at least as high as of the {\tt DD} algorithm, we can prove the conjecture of \cite{Johnson_2019} regarding the upper bound of the {\tt DD} algorithm. 

\begin{corollary}
If $m < (1-\eps) \malg$, the {\tt DD} algorithm will fail to retrieve the correct set of infected individuals \whp.
\end{corollary}
\color{black}

\subsection*{Acknowledgment}
We thank Arya Mazumdar for bringing the group testing problem to our attention.

\newpage
\begin{appendices}

\section{Notation}

\begin{center}
\begin{longtable}[h]{||p{2.1cm}|p{6.1cm}|p{6cm}||}
\hline
\textbf{Notation} & \textbf{Definition \& Properties} & \textbf{Description} \\
\hline\hline
$n$ & & population size\\
\hline
$k$ & $k\sim n^{\theta}$ for $\theta\in(0,1)$ & number of infected individuals\\
\hline
$m$& $m=c k \log (n/k)$ &number of tests\\
\hline
$x_1,\dots,x_n$ & & variable nodes\\
\hline
$V=V_n$ & $\cbc{x_i, \dots, x_n}$ & set of all individuals \\
\hline
$a_1,\dots,a_n$ & & factor nodes\\
\hline
$F=F_m$ & $\cbc{a_i, \dots, a_m}$ & set of all tests \\
\hline
$\Delta$& $\Delta=d \log(n/k)$&tests per individual, variable node degree\\
\hline
$\Gamma_1, \dots, \Gamma_m$ & $\bc{\sum_{i=1}^m \Gamma_i}/m = dn/(ck)$ & individuals per test, factor node degree\\ 
\hline
$\cG$& $(\Gamma_i)_{i \in [m]}$ & $\sigma$-algebra generated by the random variables $(\Gamma_i)_{i\in[m]}$\\
\hline
$\SIGMA \in \{0,1\}^V$ & $\sum_{i=1}^{n} \SIGMA_i=k$ &  $n$-dimensional vector of Hamming weight $k$ indicating the individuals' infection status \\
\hline
$\G=\G(n,m,\Delta)$ & & random bipartite graph on $n$ variable nodes, $m$ factor nodes and variable degree $\Delta$\\
\hline
$\partial x_i = \partial_{\G} x_i$ $\quad$ for $i \in [n]$ & $\partial x_i \subset F, \abs{\partial x_i} = \Delta$ & set of tests that individual $x_i$ participates in under $\G$\\
\hline
$\partial a_i=\partial_{\G} a_i$ $\quad$ for $i \in [m]$ & $\partial a_i \subset V, \abs{\partial a_i} = \Gamma_i$ & set of individuals in test $a_i$ under $\G$\\
\hline
$\hat\SIGMA \in \cbc{0,1}^F$ & $\hat\SIGMA_i= \vecone\cbc{\exists x\in \partial a_{i}:\SIGMA_x=1}$ & $m$-dimensional vector indicating the test outcomes \\
\hline
$\mone, \mzero$ & $\mone = \abs{\cbc{a \in F: \hat\SIGMA_a=1}}, \mzero = m - \mone$ & number of positive and negative tests \\
\hline
$V_0$& $V_0=\cbc{x\in V: \SIGMA_{x}=0}$ & set of healthy individuals \\
\hline
$V_1$& $V_1=V\setminus V_0, \abs{V}=k$ & set of infected individuals\\
\hline
$V_0^+$& $\cbc{x \in V_0:\forall a \in \partial x: \hat\SIGMA_a=1}$ & set of healthy individuals only included in positive tests\\
\hline
$V_0^-$ & $V_0^- = V_0 \setminus V_0^+$ & set of healthy individuals included in at least one negative test \\
\hline
$V_1^+$&$\cbc{x \in V_1: \forall a \in \partial x: \exists y \in \partial a \setminus \cbc{x}:\SIGMA_y=1}$  & set of infected individuals that have another infected individual in all their tests\\
\hline
$V_1^{--}$& $\cbc{x \in V_1:\exists a\in \partial x: \partial a \setminus \cbc{x} \subseteq V_0^-}$& Set of infected individuals that occur in at least one test with only healthy individuals\\
\hline
$\Gamma_{\min}, \Gamma_{\max}$& $\Gamma_{\min}=\min_{i\in[m]}\Gamma_i, \Gamma_{\max}=\max_{i\in[m]}\Gamma_i$& minimum and maximum test degree \\
\hline
$S_k(\G,\hat\SIGMA)$&  $S_k(\G,\hat\SIGMA)=\big\{\sigma\in \cbc{0,1}^V:$ $\qquad\qquad\qquad\qquad$ $\forall a_i\in[m]: \hat\SIGMA_{a_i}=\vecone\cbc{\exists x\in \partial a_i: \sigma_x=1}\big\}$& set of configurations consistent with the test results under $\G$\\
\hline
$Z_k(\G,\hat\SIGMA)$& $Z_k(\G,\hat\SIGMA)=\abs{S_k(\G,\hat\SIGMA)}$& number of configurations consistent with the test results\\
\hline
$Z_{k,\ell}(\G,\hat\SIGMA)$ & $Z_{k,\ell}(\G,\hat\SIGMA)=\abs{\cbc{\sigma\in S_k(\G,\hat\SIGMA): \scal{\SIGMA}\sigma=\ell}}$ & number of configuration consistent with the test results and with overlap $\ell$ with $\SIGMA$ \\
\hline
$\vY_i$ for $i \in [m]$ & $\vY_i= \abs{\cbc{x \in \partial a_i: \SIGMA_x = 1}}$ & number of edges that connect test $a_i$ with an infected individual \\
\hline
$\vX_i$ for $i \in [m]$ & $\vX_i \sim \Bin \bc{\Gamma_i, k/n}$ & binomially-distributed random variable with parameters $\Gamma_i$ and $k/n$ \\
\hline
$W, W'$ & $W=\sum_{i=1}^m \vecone \cbc{\vY_i = 1}, W'=\sum_{i=1}^m \vecone \cbc{\vX_i = 1}$ & $W$ is the number of tests containing a single infected individual, $W'$ is a random variable depending on $(\vX_i)_{i \in [m]}$ \\
\hline
$U$ & $U=\abs{\cbc{x \in V_1: \forall a_i \in \partial x: \vY_i>1}}$ & number of infected individuals not adjacent not any test with precisely one infected individual \\
\hline
$T$ & $\abs{\cbc{x\in V_1: \sum_{a\in\partial x}\vecone\cbc{\partial a\setminus\cbc x\subset V_0}<\delta\Delta}}$ & number of infected individuals who appear in less than $\delta\Delta$ tests as the only infected individual for some constant $\delta>0$ \\
\hline
$R$ & $R=| \big\{x \in V_1: \exists a_i \in \partial x:$ $\qquad \qquad \qquad \qquad$ $\vY_i>1 \text{ and } \partial a \setminus \cbc{x} \subset V_0 \big\}|$ & number of infected individual adjacent to some test multiple times with no other infected individual besides themselves \\
\hline
$\vA'_{i,j}, \vA'_{i,j,k}$ & & auxiliary random variables, defined in proof of \Prop~\ref{Prop_small_overlap} \\
\hline
$\cA$ & $\cA=\big\{\forall i\in[m]:$ $\qquad\qquad\qquad\qquad\qquad\qquad$ $\max_{j\in[\Gamma_i]}\vA_{i,j,1}'=\max_{j\in[\Gamma_i]}\vA'_{i,j,2}\big\}$ & event that every test under the balls-and-bins experiment features the same test result \\
\hline
$\cE$& $\cE=\cbc{\sum_{i\in[m]} \vX_i=k\Delta}$& event that the sum of $\vX_i$ is exactly $k\Delta$\\
\hline
$\cM$ & & set of all indices $i\in[m]$ for which there exists precisely one $g_i\in[\Gamma_i]$ such that $\vA_{i,g_i,1}'=1$\\
\hline
$\cN$ & & set of indices $i\in[m]$ such that $\max_{j\in[\Gamma_i]}\vA_{i,j,1}'=0$ \\
\hline
$\cR$& $\cR=\cbc{\forall x \in V_1: \abs{\cbc{a \in \partial x: \partial a \setminus \cbc{x} \subset V_0}} \geq \delta\Delta}$  & event that for every $x \in V_1$ there are at least $\delta\Delta$ tests $a\in\partial x$ for some $\delta>0$ such that $\partial a\setminus \cbc{x}\subset V_0$.\\
\hline
$\cS$ & & event that one specific $\sigma$ that has overlap $\ell$ with $\SIGMA$ belongs to $S_k(\G,\hat\SIGMA)$ \\
\hline
$\cT$ & & event that sum of independent random variable is equal to specific value, defined in \eqref{eq_def_T} \\
\hline
$\cV$ & $ \cV = \cbc{\mone = \frac{m}{2}(1 + o(1))}$ & event that around half of the tests are positive \\
\hline
$\cW$ & $\cW=$ $\qquad\qquad\qquad\qquad\qquad\qquad\qquad\qquad\qquad$ $\cbc{ |V_0^+| = (1 + o(1))(n - k)( 1 - \exp(-d/c) )^\Delta }$ & event that the size of $V_0^+$ is concentrated around its mean \\
\hline
$o(1), \omega(1)$ & & $o(1)$ [$\omega(1)$] denotes a term that vanishes [diverges] in the limit of large $n$ \\
\hline
\whp & & probability of $1-o(1)$ as $n \to \infty$ \\
\hline
\end{longtable}
\end{center}

The following sections contain the proofs of the lemmas omitted so far.

\section{Preliminaries}\label{Sec_sub_proof}
\subsection{Preliminaries}\label{Sec_stirling}
We start by stating the Chernoff bound as applied in this paper.
\begin{lemma}[Chernoff bound, \cite{JLR} (Section 2.1)] \label{lemma_chernoff}
Let $\vX \sim \Bin(n,p)$ be a binomially-distributed random variable with $\lambda = \Erw[\vX]$. Further, let
\begin{align*}
    \varphi: (-1,\infty) \to \RR_{\geq 0}, x \mapsto (1+x) \log (1+x) - x
\end{align*}
Then for some $t \geq 0$,
\begin{align*}
    \Pr \bc{\abs{\vX - \lambda} \geq t} \leq \exp \bc{- \lambda \varphi \bc{t/\lambda} - (n-\lambda) \varphi \bc{-t/(n-\lambda)}}
\end{align*}
\end{lemma}

As an application, we readily find
\begin{align*}
    \Pr \bc{\abs{\vX - \lambda} \geq \sqrt{n} \log n} \leq n^{-\omega(1)}
\end{align*}

\color{black}
	
Next, we justify that the Stirling approximation of Section \ref{Sec_InfThUpper} is accurate. Namely, let $\vA_{i,j}'=(\vA_{i,j,1}',\vA_{i,j,2}')$ be $\{0,1\}^2$-valued random variables such that $(\vA_{i,j}')_{i\in[m],j\in[\Gamma_i]}$ are mutually independent and such that
\begin{align*}
\pr\brk{\vA_{i,j}'=(1,1)}&=\ell/n,&
\pr\brk{\vA_{i,j}'=(0,1)}&=\pr\brk{\vA_{i,j}'=(1,0)}=(k-\ell)/n,\\
\pr\brk{\vA_{i,j}'=(0,0)}&=(n-2k+\ell)/n
\end{align*}
for all $i,j$.
As before, we denote by $\cT$ the event that 
\begin{align*}
\sum_{i=1}^m\sum_{j=1}^{\Gamma_i}\vecone\cbc{\vA_{i,j}'=(1,1)}&=\ell\Delta,
\qquad \sum_{i=1}^m\sum_{j=1}^{\Gamma_i}\vecone\cbc{\vA_{i,j}'=(0,0)}=(n-2k+\ell)\Delta,\\
\sum_{i=1}^m\sum_{j=1}^{\Gamma_i}\vecone\cbc{\vA_{i,j}'=(1,0)}&=
\sum_{i=1}^m\sum_{j=1}^{\Gamma_i}\vecone\cbc{\vA_{i,j}'=(0,1)}=(k-\ell)\Delta,
\end{align*}
i..e, that all of the sums on the l.h.s.\ are {\em precisely} equal to their expected values.
Since the $(\vA_{i,j}')_{i,j}$ are independent, Stirling's formula yields
\begin{align}
    \pr\brk{\cT}=\Omega \bc{(\Delta k)}^{-3/2}.  \label{eqLLTnew1}
\end{align}

	This can be seen as follows. For the sake of brevity, define  $$p_{00} = (n-2k+\ell)/n, \quad p_{11} = \ell /n, \quad \text{and} \quad p_{10} = p_{01} = (k-\ell)/n.$$
	As $\vA_{i,j}'$ is a family of independent multinomial variables $$\vA_{i,j}' \sim \Mult \bc{ 1, (p_{11}, p_{00}, p_{10}, p_{01}) },$$ we find $$\vX \sim \sum_{i=1}^m\sum_{j=1}^{\Gamma_i} \vA_{i,j}' \sim \Mult \bc{ n \Delta, (p_{11}, p_{00}, p_{10}, p_{01}) }.$$ Hence, the probability of event $\cT$ occurring is the probability, that $\vX$ hits its expectation.
	Thus, using the very basic approximation $n! = \Theta \bc{ \sqrt{n}} (n/\eul{})^n$ we find 
	\begin{align}
	    \notag \Pr \bc{ \cT } &= \frac{(n \Delta)! (\ell/n)^{\ell \Delta} ((n-2k + \ell)/n)^{(n-2k+\ell) \Delta} ((k-\ell)/n)^{2 (k - \ell) \Delta}}{ (\ell \Delta)! ( (n-2k + \ell)\Delta )! ( (k-\ell)\Delta )! ( (k-\ell)\Delta )!} \\
	    \notag &= \Theta \bc{ \frac{ \sqrt{n \Delta} }{ \sqrt{ \ell (n - 2k \ell) (k - \ell)^2 \Delta^4  } } } \bc{ \frac{ (n \Delta)^n (\ell/n)^\ell ((n - 2k + \ell)/n)^{n-2k + \ell} ((k-\ell)/n)^{2(k - \ell)}} { \ell^\ell (n - 2k + \ell)^{n - 2k + \ell} (k -\ell)^{2(k - \ell)}} }^{\Delta} \\
	    \notag &= (1 + O(1/n)) \Theta \bc{ \frac{ \sqrt{n} }{ \sqrt{n} \sqrt{ \ell k^2 - 2 \ell^2 k + \ell^3 - k \bc{ 2 \ell k^2 / n - 2 k \ell^2 / n } + \ell^4/n } } } \\
	    &= \Omega \bc{ \sqrt{\Delta^{-3} (\ell k^2 + \ell^2 k + \ell^3)^{-1} } } = \Omega \bc{ (\Delta k)^{-3/2} }, \label{eq_explain_omega}
	\end{align}
	where \eqref{eq_explain_omega} follows immediately from $\ell \leq k = o(n)$ and directly implies \eqref{eqLLTnew1}. In due course we apply similar calculations often, some calculations involve conditional probabilities. These conditions are only restricting $\Gamma_i$ to take specific (common) values and clearly the above argument is totally invariant under different values of $\Gamma_i$, as long as $\sum_i^m \Gamma_i = n \Delta$. 
	
\color{black}

\subsection{Getting started}
In the next step, recall that neighbourhoods of different tests in the random multi-graph seizably intersect. To
cope with the ensuing correlations, we introduce a new family of random variables that, as we will see, are closely
related to the statistics of the appearances of infected/uninfected individuals in the various tests.
Specifically, recalling that $\Gamma_i$ signifies the degree of test $a_i$ \mhk{and that $\sum_{i=1}^m \Gamma_i = n \Delta$}, let $(\vX_i)_{i \in [m]}$ be a sequence of independent $\Bin(\Gamma_i, k/n)$-variables. Moreover, let
\begin{align*}
	\cE&=\cbc{\sum_{i\in[m]}\vX_i=k\Delta}.
\end{align*}
Because the $\vX_i$ are mutually independent, \pl{Stirling's formula} shows that
\begin{align}\label{eqEprob}
	\pr\brk\cE&=\Omega(1/\sqrt{\Delta k}),
\end{align}
which follows along the lines of \Sec~\ref{Sec_stirling}.
Additionally, let $\vY_i$ be the number of edges that connect test $a_i$ with an infected individual.
(Since $\G$ is a multi-graph, it is possible that an infected individual contributes more than one to $\vY_i$.) Further, let $\cG$ be the $\sigma$-algebra generated by the random variables $(\Gamma_i)_{i\in[m]}$. \mhk{Whenever we condition on $\Gamma$, we assume that the bounds from \Lem~\ref{Lemma_GammaMinMax} and \ref{Lemma_m0} hold.}

\begin{lemma}\label{Lemma_Elemma}
	Given $\cG$, the vectors $(\vY_1,\ldots,\vY_{m})$ and $(\vX_1,\ldots,\vX_m)$ given $\cE$ are identically distributed.
\end{lemma}

\begin{proof}
	For any integer sequence $(y_i)_{i\in[m]}$ with $y_i\geq0$ and $\sum_{i\in[m]}y_i=k\Delta$ we have
	\begin{align*}
		\pr\brk{\forall i\in[m]:\vY_i=y_i\mid\cG}&=\frac{\binom{k\Delta}{y_1,\ldots,y_m}\binom{(n-k)\Delta}{\Gamma_1-y_1,\ldots,\Gamma_m-y_m}}{\binom{n\Delta}{\Gamma_1,\ldots,\Gamma_m}}=\mhk{\frac{\prod_{i=1}^m \frac{\Gamma_i!}{y_i! (\Gamma_i-y_i)!}}{\frac{(n\Delta)!}{(k\Delta)! ((n-k)\Delta)!}}}=\binom{n\Delta}{k\Delta}^{-1}\prod_{i=1}^m\binom{\Gamma_i}{y_i}.
	\end{align*}
	Hence, for any sequences $(y_i),(y_i')$ we obtain
	\begin{align*}
		\frac{\pr\brk{\forall i\in[m]:\vY_i=y_i\mid\cG}}{\pr\brk{\forall i\in[m]:\vY_i=y_i'\mid\cG}}
		&=\prod_{i=1}^m\frac{\binom{\Gamma_i}{y_i}}{\binom{\Gamma_i}{y_i'}}
		=\frac{\pr\brk{\forall i\in[m]:\vX_i=y_i\mid\cG,\cE}}{\pr\brk{\forall i\in[m]:\vX_i=y_i'\mid\cG,\cE}},
	\end{align*}
	as claimed.
\end{proof}

\subsection{Proof of \Lem~\ref{Lemma_GammaMinMax}}
Since each variable draws a sequence of $\Delta$ tests uniformly at random, for every $i\in[m]$ the degree $\Gamma_i$ has distribution $\Bin(n\Delta,1/m)$.
Therefore, the assertion follows from the Chernoff bound.

\subsection{Proof of \Lem~\ref{Lemma_m0}}
Let $\mzero'=\sum_{i=1}^m\vecone\cbc{\vX_i=0}$.
Then $\Erw[\mzero']=\sum_{i=1}^m\pr\brk{\Bin(\Gamma_i,k/n))=0}=\sum_{i=1}^m(1-k/n)^{\Gamma_i}$.
Hence, \Lem~\ref{Lemma_GammaMinMax} shows that with probability $1-o(n^{-2})$,
\begin{align}\label{eqLemma_m0_1}
	\Erw[\mzero'\mid\cG]&\geq m(1-k/n)^{\Gamma_{\max}}=m\exp\bc{(\Delta n/m+O(\sqrt{\Delta n/m}\log n))\log(1-k/n)}\\
	&=m\bc{\exp(-d/c)+O(\sqrt{k/n}\log n)},\\
	\Erw[\mzero'\mid\cG]&\leq m(1-k/n)^{\Gamma_{\min}}=m\bc{\exp(-d/c)+O(\sqrt{k/n}\log n)}.
\end{align}
Because the $\vX_i$ are mutually independent, $\mzero'$ is a binomial variable.
Therefore, the Chernoff bound \pl{(e.g. \Lem~\ref{lemma_chernoff})} shows that
\begin{align}\label{eqLemma_m0_2}
	\pr\brk{\abs{\mzero'-\Erw[\mzero'\mid\cG]}>\sqrt{m}\log n\mid\cG}=o(n^{-10}).
\end{align}
Finally, the assertion follows from \eqref{eqEprob}, \eqref{eqLemma_m0_1}--\eqref{eqLemma_m0_2} and \Lem~\ref{Lemma_Elemma}.

\subsection{Proof of \Lem~\ref{Lemma_Delta_too_small_big}}
The expected degree of a test $a_i$ equals $\Delta n/m$.
Therefore, if $\Delta=o(\log(n/k))$, then by \Lem~\ref{Lemma_m0}, $\mone=o(m)$ \whp{}
To exploit this fact, call $\sigma\in\cbc{0,1}^V$ of Hamming weight $k$ {\em bad} for $\G$ if given $\SIGMA=\sigma$ we indeed have $\mone=o(m)$.
Let $B(\G)$ be the set of all such bad $\sigma$.
Then \whp\ $\G$ has the property that $|B(\G)|\sim\binom nk$, \pl{i.e. asymptotically most configurations will have few positive tests}.
Now, condition on the event that $|B(\G)|\sim\binom nk$ and let $\cB$ be the set of all subsets of $[m]$ of size $o(m)$.
Further, let $f_{\G}:B(\G)\to\cB$ map $\sigma\in\cbc{0,1}^V$ to the corresponding set of positive tests.
Finally, let $B'(\G)$ be the set of all $\sigma\in B(\G)$ such that $|f_{\G}^{-1}(f_{\G}(\sigma))|<n$, \pl{i.e. the set of all configurations for which there are less than $n$ other configurations rendering the same test results}.
Then
\begin{align*}
	|B'(\G)|&\leq n|\cB|\leq n\binom{m}{o(m)}=\exp(o(m))=o\bc{\binom nk}.
\end{align*}
Consequently, \whp\ over the choice of $\G$ and $\SIGMA$ we have $Z_k(\G,\hat\SIGMA)\geq n$.
The same argument applies for $\log(n/k)=o(\Delta)$ with the term `positive test' replaced by `negative test'.

\subsection{Proof of \Prop~\ref{Lemma_V_all}} \label{subsec_start_prop}
We start by proving part (\ref{Lemma_V0-}) \mhk{using a straightforward second-moment calculation}. \mhk{Recall $\Delta=d \log(n/k)$ and $m=ck\log(n/k)$.}
	\Lem~\ref{Lemma_GammaMinMax} and \Lem~\ref{Lemma_m0} show that with probability at least $1-o(n^{-2})$ the total degree of the negative tests comes to
	\begin{align*}
		\sum_{i=1}^m\vecone\cbc{\partial a_i\subset V_0}\Gamma_i&=\Delta n\exp(-d/c)+O\bc{\sqrt{m}\log^2 (n)\Delta n/m+m\sqrt{\Delta n/m}\log n}\\
		&=\Delta n\exp(-d/c)+O\bc{\bc{\sqrt{nk}+n/\sqrt{k}}\log^3n}=\Delta n\bc{\exp(-d/c)+n^{-\Omega(1)}}.
	\end{align*}
	Consequently, with probability at least $1-o(n^{-2})$ the total number of edges between $V_0$ and the set of positive tests is $\Delta n\bc{1-\exp(-d/c)+n^{-\Omega(1)}}$. \pl{Moreover, the total number of edges between $V_0$ and all tests comes down to $\Delta(n-k)$}.
	Given these events \pl{and since each individual is assigned to tests uniformly at random with replacement}, the probability that a given $x\in V_0$ belongs to $V_0^+$ comes out as
	\begin{align*}
		\binom{\Delta n\bc{1-\exp(-d/c)+n^{-\Omega(1)}}}{\Delta}\binom{\Delta(n-k)}{\Delta}^{-1}&
		=\bc{1+n^{-\Omega(1)}}\bc{1-\exp(-d/c)}^\Delta.
	\end{align*}
	\pl{Next}, we estimate the probability that $x,x'\in V_0$ both belong to $V_0^+$:
	\begin{align*}
		\binom{\Delta n\bc{1-\exp(-d/c)+n^{-\Omega(1)}}}{2\Delta}\binom{\Delta(n-k)}{2\Delta}^{-1}&
		=\bc{1+n^{-\Omega(1)}}\bc{1-\exp(-d/c)}^{2\Delta},
	\end{align*}
	Hence, $\Erw[|\zeroplus{}|^2\mid\cG]-\Erw[|\zeroplus{}|\mid\cG]^2=O(n^{2-\Omega(1)})$.
	Therefore, the assertion follows from Chebyshev's inequality.

Proceeding with part (\ref{Lemma_V1plus}), let \pl{the number of tests containing a single infected individual be}
	\begin{align*}
		W&=\sum_{i=1}^m\vecone\cbc{\vY_i=1},&W'&=\sum_{i=1}^m\vecone\cbc{\vX_i=1}.
	\end{align*}
	Then \Lem~\ref{Lemma_GammaMinMax} shows that \whp\ 
	\begin{align*}
		\Erw[W']&=\sum_{i=1}^m\frac{\Gamma_i k}{n}\bc{1-k/n}^{\Gamma_i-1}
		\leq \frac{\Gamma_{\max} km}{n}\bc{1-k/n}^{\Gamma_{\min}-1}\\
		&=\bc{1+n^{-\Omega(1)}}k\Delta(1-k/n)^{\Delta n/m}=\bc{1+n^{-\Omega(1)}}k\Delta\exp(-d/c)
	\end{align*}
	Analogously,
	\begin{align*}
		\Erw[W']&\geq \frac{\Gamma_{\min} km}{n}\bc{1-k/n}^{\Gamma_{\max}}
		=\bc{1+n^{-\Omega(1)}}k\Delta\exp(-d/c).
	\end{align*}
	Hence, because $W'$ is a binomial random variable, the Chernoff bound \pl{(e.g. \Lem~\ref{lemma_chernoff})} shows that
	$$\pr\brk{W'=\bc{1+n^{-\Omega(1)}}k\Delta\exp(-d/c)\mid\cG}=1-o(n^{-9}).$$
	Therefore, \eqref{eqEprob} yields
	\begin{align}\label{eqLemma_V1plus_10}
		\pr\brk{W=\bc{1+n^{-\Omega(1)}}k\Delta\exp(-d/c)\mid\cG} &=1-o(n^{-7}).
	\end{align}
	Now, let $U$ be the number of $x\in V_1$ that are not adjacent to any test with precisely one positive individual. \mhk{An individual $x \in V_1$ counts towards $U$, if out of all possible assignment $k\Delta$, it is only assigned to those tests where it is not the only infected individual (there are a total of $k\Delta-W$ such assignments).}
	\mhk{Using the notation $n^{\underline{k}} = n (n-1) \dots (n-k+1) $} and recalling $\Delta=\Theta(\log n)$, the bound on $W$ yields
	\color{black}
	\begin{align*}
		\Erw[U\mid\cG,W]&=k\binom{k\Delta-W}\Delta\binom{k\Delta}\Delta^{-1} = k \frac{(k\Delta-W)^{\underline{\Delta}}}{(k\Delta)^{\underline{\Delta}}} = \bc{1+n^{-\Omega(1)}} k \bc{\frac{k\Delta - W}{k \Delta}}^\Delta\\
		&=\bc{1+n^{-\Omega(1)}}k(1-W/k\Delta)^\Delta=\bc{1+n^{-\Omega(1)}}k(1-\exp(-d/c))^\Delta.
	\end{align*}
	\color{black}
	By a similar token we obtain
	\begin{align*}
		\Erw[U^2\mid\cG,W]&=k^2\binom{k\Delta-W}{2\Delta}\binom{k\Delta}{2\Delta}^{-1}=\bc{1+n^{-\Omega(1)}}\Erw[U\mid\cG,W]^2.
	\end{align*}
	Therefore, Chebyshev's inequality shows that \whp{}
	\begin{align}\label{eqLemma_V1plus_1}
		U&=\bc{1+n^{-\Omega(1)}}k(1-\exp(-d/c))^\Delta.
	\end{align}	
	To complete the proof we need to compare $U$ and $\abs\oneplus{}$.
	Clearly, $U\geq\abs\oneplus{}$.
	But the inequality may be strict because $U$ includes positive individuals that appear twice in the same test. \pl{To be precise, an individual might be assigned to one test twice as the only infected individual. Such an individual should not be in $V_1^+$, but it shows up in $U$.}
	Indeed, letting $R$ be the number of such individuals, we obtain $\abs\oneplus{}\geq U-R$.
	Hence, we are left to estimate $R$.
	\pl{To this end, we observe that the probability that an individual appears in a specific test twice is upper-bounded by $\bc{\Delta/m}^2$. \mhk{Recall $m=ck\log(n/k)$ and $\Delta=d \log(n/k)$.} Consequently, taking the union bound over all tests and infected individuals we yield
	$$\Erw[R\mid\cG]\leq k m \bc{\frac{\Delta}{m}}^2 = O(\log n).$$}
	Since by assumption the r.h.s.\ of \eqref{eqLemma_V1plus_1} is $n^{\Omega(1)}$, we conclude that 
	$\abs\oneplus{}\geq U-R=n^{\Omega(1)}$ \whp, as claimed.

Next, we consider (\ref{Lemma_V1plus_reverse}). Define $U$ as in the proof of \Prop~\ref{Lemma_V_all}(\ref{Lemma_V1plus}).
	Then we know that $U\geq\abs\oneplus{}$.
	Hence, if $k(1-\exp(-d/c))^\Delta=o(1)$ then $\abs{V_1^+}=o(1)$ due to \eqref{eqLemma_V1plus_1}.

For part (\ref{Lemma_V++}), we observe for a given $c$ that $\min_d (1-\exp(-d/c))^\Delta$ is attained at $d=c\log 2$. To see this, consider the function $f(d)=(1-\exp(-d/c))^\Delta=n^{(1-\theta)d \log (1-\exp(-d/c))}$ and observe that the minimum of $f(d)$ coincides with the minimum of $g(d)=d \log (1-\exp(-d/c))$. Letting $x=d/c$, the derivatives read as
	\begin{align*}
	g(x) &= cx \log (1-\exp(-x)) \\
	g'(x) &= c \bc{\log (1-\exp(-x)) + \frac{x \exp(-x)}{1-\exp(-x)}} \\
	g''(x) &= c \bc{-\frac{(x-2) \exp(x)+2}{(\exp(x)-1)^2}}
	\end{align*}		
For \pl{$d>0$}, the unique maximum is attained at $x=\log 2$ and accordingly, $d=c \log 2$. Furthermore, it is the case that $k(1-\exp(-\log 2))^{c \log 2 \log(n/k)} \geq n^{\Omega(1)}$ and therefore by \Prop~\ref{Lemma_V_all}(\ref{Lemma_V1plus}), $\abs\oneplus{} = n^{\Omega(1)}$. By a similar token by \Prop~\ref{Lemma_V_all}(\ref{Lemma_V0-}), $\abs\zeroplus{} = n^{\Omega(1)}$.
 
Finally, for part (\ref{Lemma_V+_reverse}), setting $d=c\log 2$, we see that $k(1-\exp(-\log 2))^{c \log 2 \log(n/k)} =o(1) $ and therefore by \Prop~\ref{Lemma_V_all}(\ref{Lemma_V1plus_reverse}), $\abs\oneplus{} = o(1)$.

\section{The information-theoretic upper bound}

\subsection{Proof of \Lem~\ref{Lemma_big_overlap}} \label{sec_big_overlap}
	The term $\bink{k}{\ell}\bink{n-k}{k-\ell}$ accounts for the number of assignments $\sigma\in\cbc{0,1}^V$ of Hamming weight $k$ whose overlap with $\SIGMA$ is equal to $\ell$.
	Hence, with $\cS$ {being} the event that one specific $\sigma\in\{0,1\}^{V}$ that has overlap $\ell$ with $\SIGMA$ belongs to $S_{k,\ell}(\G,\hat\SIGMA)$, we need to show that
	\begin{align}\label{eqLemma_small_overlap1}
	\pr\brk{\cS\mid\cG,\cR,\mzero}&\leq
	O\bc{\bc{\Delta k}^{3/2}} \cdot\bc{1-\bc{1-\frac{k-\ell}{n-k}}^{\Gamma_{\max}}}^{\delta \Delta (k-\ell)}\bc{\frac{n-2k+\ell}{n-k}}^{\Gamma_{\min}\mzero}
	\end{align}
	Due to symmetry we may assume that $\SIGMA_{x_i}=\vecone\{i\leq k\}$ and that $\sigma_{x_i}=\vecone\{i\leq\ell\}+\vecone\{k<i\leq 2k-\ell\}$.
	
	Proceeding as in the proof of \Prop~\ref{Prop_small_overlap}, we think of each test $a_i$ as a bin of capacity $\Gamma_i$ and of each clone $(x_i,h)$, $h\in[\Delta]$, of an individual as a ball labelled $(\SIGMA_{x_i},\sigma_{x_i})\in\{0,1\}^2$.
	We toss the $\Delta n$ balls randomly into the bins.
	For $i\in[m]$ and for $j\in[\Gamma_i]$ we let $\vA_{i,j}=(\vA_{i,j,1},\vA_{i,j,2})\in\{0,1\}^2$ be the label of the $j$th ball that ends up in bin number $i$.
	To cope with this experiment we introduce a new set $\{0,1\}^2$-valued random variables $\vA_{i,j}'=(\vA_{i,j,1}',\vA_{i,j,2}')$ such that $(\vA_{i,j}')_{i\in[m],j\in[\Gamma_i]}$ are mutually independent and
	\begin{align*}
	\pr\brk{\vA_{i,j}'=(1,1)}&=\ell/n,&
	\pr\brk{\vA_{i,j}'=(0,1)}&=\pr\brk{\vA_{i,j}'=(1,0)}=(k-\ell)/n,\\
	\pr\brk{\vA_{i,j}'=(0,0)}&=(n-2k+\ell)/n
	\end{align*}
	for all $i,j$.
	With $\cT$ being the event that 
	\begin{align} \label{eq_def_U}
	\sum_{i=1}^m\sum_{j=1}^{\Gamma_i}\vecone\cbc{\vA_{i,j}'=(1,1)}&=\ell\Delta,
	\qquad \sum_{i=1}^m\sum_{j=1}^{\Gamma_i}\vecone\cbc{\vA_{i,j}'=(0,0)}=(n-2k+\ell)\Delta,\\
	\sum_{i=1}^m\sum_{j=1}^{\Gamma_i}\vecone\cbc{\vA_{i,j}'=(1,0)}&=
	\sum_{i=1}^m\sum_{j=1}^{\Gamma_i}\vecone\cbc{\vA_{i,j}'=(0,1)}=(k-\ell)\Delta,
	\end{align}
	the vector $\vA'=(\vA_{i,j}')_{i,j}$ given $\cT$ is distributed as $\vA=(\vA_{i,j})_{i,j}$ given $\cG$.
	\mhk{Moreover, with similar arguments as in \Sec~\ref{Sec_stirling},} \pl{Stirling's formula} yields
	\begin{align}\label{eqLLT8}
	\pr\brk{\cT}=\Omega((\Delta k)^{-3/2}).
	\end{align}
	\color{black}
	Let $\cN$ be the set of indices $i\in[m]$ such that $\max_{j\in[\Gamma_i]}\vA_{i,j,1}'=0$.
	Moreover, let $\cM$ be the set of all indices $i\in[m]$ for which there exists precisely one $g_i\in[\Gamma_i]$ such that $\vA_{i,g_i,1}'=1$
	and such that for this index we have $\vA_{i,g_i,2}'=0$.
	Further, let
	\begin{align*}
	\cS'&=\cbc{\forall i\in\cN:\max_{j\in[\Gamma_i]}\vA_{i,j,2}'=0},&
	\cS''&=\cbc{\forall i\in\cM:\max_{j\in[\Gamma_i]}\vA_{i,j,2}'=1}.
	\end{align*}
	Then
	\begin{align*}
	\cA=\cbc{\forall i\in[m]:\max_{j\in[k]}\vA_{i,j,1}'=\max_{j\in[k]}\vA'_{i,j,2}}\subset\cS' \cap\cS''.
	\end{align*}
	Furthermore, given $\cN,\cM$ the events $\cS',\cS''$ are independent and
	\begin{align*}
	\pr\brk{\cS'\mid\cN}&= \prod_{i \in \cN}\bcfr{n-2k+\ell}{n-k}^{\Gamma_i} \leq \bcfr{n-2k+\ell}{n-k}^{\Gamma_{\min}\abs\cN},\\
	\pr\brk{\cS''\mid\cM}&=\prod_{i\in \cM} \bc{1-\bc{1-\frac{k-\ell}{n-k}}^{\Gamma_{i}-1}}\leq
	\bc{1-\bc{1-\frac{k-\ell}{n-k}}^{\Gamma_{\max}}}^{\abs\cM}.
	\end{align*}

\pl{For an intuitive explanation of the above expressions, please refer to the section immediately following the statement of the \Lem~\ref{Lemma_big_overlap}}.
	Given $\abs\cN\geq \bc{1-n^{-\Omega(1)}}\mzero$ and $\abs\cM\geq\delta\Delta(k-\ell)$, we obtain
	\begin{align}\label{eqLLT9}
	\pr\brk{\cA\mid \abs\cN\geq \bc{1-n^{-\Omega(1)}}\mzero,\abs\cM\geq\delta\Delta(k-\ell)}&\leq 
	\bcfr{n-2k+\ell}{n-k}^{\Gamma_{\min}\mzero}
	\bc{1-\bc{1-\frac{k-\ell}{n-k}}^{\Gamma_{\max}}}^{\delta\Delta(k-\ell)}.
	\end{align}

	\color{black}
	Moreover, we find by \ref{Lemma_rigid}, the concentration of $\abs \cN$ and the fact that $\Erw\brk{\abs{\cN}}=\Erw\brk{\mzero}=m/2$
	\begin{align*}
	    \Pr\bc{\abs\cN\geq \bc{1-n^{-\Omega(1)}}\mzero,\abs\cM\geq\delta\Delta(k-\ell)} = 1-o(1)
	\end{align*}
	and thus 
	\begin{align*}
	    \pr\brk{\cT | \abs\cN\geq \bc{1-n^{-\Omega(1)}}\mzero,\abs\cM\geq\delta\Delta(k-\ell)}=\Omega((\Delta k)^{-3/2}).
	\end{align*}
	\color{black}
	Combining \eqref{eqLLT8}--\eqref{eqLLT9} and using the trivial bound
	\begin{align}\label{eqLLT10}
	\pr\brk{\cT\mid \cS,\cS',\abs\cN\geq \bc{1-n^{-\Omega(1)}}\mzero,\abs\cM\geq\delta\Delta(k-\ell)}&\leq 1,
	\end{align}
	we obtain \mhk{by Bayes Theorem}
	\begin{align}\label{eqLLT11}
	\pr\brk{\cA\mid\cT,\abs\cN\geq \bc{1-n^{-\Omega(1)}}\mzero,\abs\cM\geq\delta\Delta(k-\ell)}
	&\leq O\bc{\bc{\Delta k}^{3/2}} \bcfr{n-2k+\ell}{n-k}^{\bc{1-n^{-\Omega(1)}} \Gamma_{\min}\mzero} \bc{1-\bc{1-\frac{k-\ell}{n-k}}^{\Gamma_{\max}}}^{\delta\Delta(k-\ell)}.
	\end{align}
	Because $\vA'=(\vA_{i,j}')_{i,j}$ given $\cT$ is distributed as $\vA=(\vA_{i,j})_{i,j}$ given $\cG$, \eqref{eqLemma_small_overlap1} follows from \eqref{eqLLT11}.

\section{The {\tt SCOMP} algorithm}



\color{black}
\subsection{Proof of \Prop~\ref{Lemma_--}} \label{Proof_of_Lemma_--}
\color{black}
The proof of \Prop~\ref{Lemma_--} proceeds in three steps. First, we show that $\abs{V_0^+}$ is concentrated around its expectation. $\cW$ denotes the corresponding event. Second, we need to get a handle on the subtle dependencies in $\G$. To this end, we introduce a set of independent multinomial random variables indexed over the tests. Whereas $\vY^i_1,\vY^i_{0+},\vY^i_{0-}$ denotes the number of infected, potentially false positive and definitively healthy individuals in test $a_i$, respectively, the triple $(\vX^i_1,\vX^i_{0+},\vX^i_{0-})$ denote the corresponding multinomial random variable. We will show that conditioned on the sum of $\vX^i_1,\vX^i_{0+},\vX^i_{0-}$ hitting the total number of individuals of the three types, $(\vX^i_1,\vX^i_{0+},\vX^i_{0-})$ is distributed like $\vY^i_1,\vY^i_{0+},\vY^i_{0-}$. The technical workout is delicate, but is based on standard results from balls-into-bins experiments. Third, we show that for $m<(1-\eps) \malg$, the number of tests $W$ for which $\vX^i_1=1$ and $\vX^i_{0+}=0$ decays exponentially in $n$, which implies that $\oneminusminus = \emptyset$ \whp

\color{black}

\begin{proof}
\Lem~\ref{Lemma_Delta_too_small_big} implies that the optimal choice for the variable degree is $\Delta = d \log(n/k)$ for a constant $d$.
Let $\mone$ be the amount of positive tests and, w.l.o.g. assume that $a_1 ... a_{\mone}$ are the positive tests and define 
\begin{align*}
\cW &= \cbc{ |V_0^+| = (1 + o(1))(n - k)( 1 - \exp(-d/c) )^\Delta }.
\end{align*}
as the event that the number of `potential false positives' $\abs{V_0^+}$ is highly concentrated around its mean.
Then by 
\Prop~\ref{Lemma_V_all}(\ref{Lemma_V0-}), we find
\begin{align}
\Pr[ \cW ] \geq 1 - o(1)
\end{align}
Similarly as before, we introduce a family of independent random variables corresponding to the tests.

Let $\vY_1^1,\ldots,\vY_1^{\mone}$ be the number of ones \mhk{in the tests corresponding to} $a_1,\ldots,a_{\mone}$ \mhk{respectively}.
Let $\vY_{0+}^1,\ldots,\vY_{0+}^{\mone}$ count the $\zeroplus$ occurrences in $a_1,\ldots,a_{\mone}$.
Let $\vY_{0-}^1,\ldots,\vY_{0-}^{\mone}$ count the $\zerominus$ occurrences in $a_1,\ldots,a_{\mone}$. By definition we find $\vY_{0-}^i = \Gamma_i -\vY_{0+}^i - \vY_{1}^i$. We introduce auxiliary variables $\vX_1^1,\ldots,\vX^{\mone}_{1}$, $\vX^1_{0+},\ldots,\vX^{\mone}_{0+}, \vX^1_{0-},\ldots,\vX^{\mone}_{0-}$ such that $(\vX^i_1, \vX^i_{0+}, \vX^i_{0-})$ have distribution $$\Mult_{\geq (1,0,0)}(\Gamma_i, p, q, 1-p-q),$$ a multinomial distribution conditioned on the first variable being at least one. The triples $\left( (X^i_1, X^i_{0^+}, X^i_{0-}) \right)_{i \in \mone}$ are mutually independent. We \mhk{seek a choice of $p$ satisfying the equation} 
\begin{align*}
p := \frac{k \Delta}{\sum_{i=1}^{\mone} \frac{\Gamma_i}{1-(1-p)^{\Gamma_i}}} \qquad \qquad \text{and} \qquad \qquad q := \frac{\abs{\zeroplus} \Delta}{\sum_{i=1}^{\mone} \frac{\Gamma_i}{1-(1-p)^{\Gamma_i}}}.
\end{align*}
\mhk{and will show following equation \eqref{W_second} that such a choice exists.}
Define 
\begin{align*}
\cE = \cbc{\sum_{i=1}^{\mone}\vX^i_1=k\Delta,\sum_{i=1}^{\mone}\vX^i_{0+}=|\zeroplus|\Delta}.
\end{align*}
\mhk{Along the lines of \Sec~\ref{Sec_stirling} ,} \pl{Stirling's formula} implies
\begin{align}
\pr\brk{\cE}=\Omega(1/n). \label{Pr_cE_Poly}
\end{align}
Moreover, $(\vY^1_1,\vY^1_{0+},\vY^1_{0-},\ldots,\vY^{\mone}_1,\vY^{\mone}_{0+}, \vY^{\mone}_{0-})$ and $(\vX^1_1,\vX^1_{0+},\vX^1_{0-},\ldots,\vX^{\mone}_1,\vX^{\mone}_{0+}, \vX^{\mone}_{0-})$ given $\cE$ are identically distributed. This can be seen as follows: 
\begin{align*}
\notag \Pr & \left[ \forall i \in [\mone]: (\vY^{i}_1,\vY^{i}_{0+}, \vY^{i}_{0-}) = (y_i, y'_i, y''_i) \mid \cG, |V_0^+|, \mone \right] \\
& = \frac{\binom{k \Delta}{y_1 \dots y_{\mone}} \binom{|V_0^+| \Delta}{y_1' \dots y_{\mone}'} \binom{\sum_{i=1}^{\mone} \Gamma_i- \bc{k+|V_0^+|}\Delta}{\Gamma_1 - y_1 - y_1', \dots, \Gamma_{\mone} - y_{\mone} - y'_{\mone}}}{\binom{ \sum_{i=1}^{\mone} \Gamma_i }{ \Gamma_1, \dots , \Gamma_{\mone} }} \ind{\forall i \in [\mone]: y''_i = \Gamma_i - y_i - y'_i} \\
& = \binom{\sum_{i=1}^{\mone} \Gamma_i}{k\Delta, |V_0^+|\Delta, \sum_{i=1}^{\mone} \Gamma_i-\bc{k+|V_0^+|}\Delta} \prod_{i=1}^{\mone} \binom{\Gamma_i}{y_i, y'_i, \Gamma - y_i -y'_i} \ind{\forall i \in [\mone]: y''_i = \Gamma_i - y_i - y'_i}.
\end{align*}
Thus, given $y''_i = \Gamma_i - y_i - y'_i$ and $\tilde{y}''_i = \Gamma_i - \tilde{y}_i - \tilde{y}'_i$ for all $i \in [\mone]$, we find
\begin{align}
& \frac{\Pr \left[ \forall i \in [\mone]: (\vY^1_1,\vY^1_{0+},\vY^1_{0-})
= (y_i, y'_i, y'') \mid \cG, |V_0|^+, \mone \right]}{\Pr \left[ \forall i \in [\mone]: (\vY^1_1,\vY^1_{0+},\vY^1_{0-}) = (\tilde{y}_i, \tilde{y}'_i, \tilde{y}''_i) \mid \cG, |V_0|^+, \mone \right]} = \prod_{i=1}^{\mone} \frac{\binom{\Gamma_i}{y_i, y'_i, \Gamma - y_i -y'_i}}{\binom{\Gamma_i}{\tilde{y}_i, \tilde{y}'_i, \Gamma - \tilde{y}_i -\tilde{y}'_i}}. \label{Eq_XiYiSame_1}
\end{align}
Given $x_i'' = \Gamma_i - x_i - x'_i$, we find:
\begin{align*}
\Pr & \left[ \forall i \in [\mone]: (\vX^1_1,\vX^1_{0+},\vX^1_{0-}) = (x_i, x'_i, x_i'') \mid \cE, \cG, |V_0|^+, \mone \right] \\
&= \prod_{i=1}^{\mone} \binom{\Gamma_i}{x_i, x_i', x_i''} p^{x_i} q^{x_i'} (1-p-q)^{x''_i} \frac{1}{1 - (1-p)^{\Gamma_i}} \\ 
& = p^{k\Delta} q^{|V_0^+|\Delta} (1-p-q)^{\sum_{i=1}^{\mone} \Gamma_i-\Delta(k+|V_0^+|)} \prod_{i=1}^{\mone} \frac{1}{1 - (1-p)^{\Gamma_i}} \binom{\Gamma_i}{x_i, x_i', x_i''} 
\end{align*}
where the last equality follows from the fact that we conditioned on $\cE$. Since the first terms are independent of $x_i, x'_i, x''_i$, we find
\begin{align*}
\frac{\Pr \left[ \forall i \in [\mone]: (\vX^i_1,\vX^i_{0+},\vX^i_{0-}) = (x_i, x'_i, x_i'') \mid \cE, \cG, |V_0|^+, \mone \right]}{\Pr \left[ \forall i \in [\mone]: (\vX^i_1,\vX^i_{0+},\vX^i_{0-}) = (\tilde{x}_i, \tilde{x}_i', \tilde{x}_i'') \mid \cE, \cG, |V_0|^+, \mone \right]} = \prod_{i=1}^{\mone} \frac{\binom{\Gamma_i}{x_i, x'_i, \Gamma - x_i -x'_i}}{\binom{\Gamma_i}{\tilde{x}_i, \tilde{x}'_i, \Gamma - \tilde{x}_i -\tilde{x}'_i}}.
\end{align*}
Therefore, given $\Gamma_i = x_i + x'_i + x_i'' = \tilde{x}_i + \tilde{x}'_i + \tilde{x}''_i,$ we have by comparison with \eqref{Eq_XiYiSame_1},
\begin{align*}
&\frac{\Pr \left[ \forall i \in [\mone]: (\vX^i_1,\vX^i_{0+},\vX^i_{0-}) = (x_i, x'_i, x_i'') \mid \cE, \cG, |V_0|^+, \mone \right]}{\Pr \left[ \forall i \in [\mone]: (\vX^i_1,\vX^i_{0+},\vX^i_{0-}) = (\tilde{x}_i, \tilde{x}_i', \tilde{x}_i'') \mid \cE, \cG, |V_0|^+, \mone \right]} \\
= &\frac{\Pr \left[ \forall i \in [\mone]: (\vY^i_1,\vY^i_{0+},\vY^i_{0-})
= (x_i, x'_i, x_i'') \mid \cG, \mone \right]}{\Pr \left[ \forall i \in [\mone]: (\vY^i_1,\vY^i_{0+},\vY^i_{0-}) = (\tilde{x}_i, \tilde{x}'_i, \tilde{x}''_i) \mid \cG, \mone \right]},
\end{align*}
which yields the claim. Let
\begin{align*}
W&=\sum_{i=1}^{\mone}\vecone\cbc{\vX_1^i+\vX_{0+}^i=1}.
\end{align*}
\pl{be the number of positive tests that contain exactly one infected individual and no healthy individuals in $V_0^+$. Note that this split is the only possibility for the test to be positive.}
Then
\begin{align*}
\Erw[W | \cG, \cE, \abs{\zeroplus}, \mone] = \sum_{i=1}^{\mone} \Pr[\vX_{1}^i=1, \vX_{0+}^i=0, \vX_{0-}^i=\Gamma_i-1]=\sum_{i=1}^{\mone}\frac{\Gamma_i p (1-p-q)^{\Gamma_i-1}}{1-(1-p)^{\Gamma_i}}.
\end{align*}
By \Lem~\ref{Lemma_m0} we readily find for any choice of $c,d=\Theta(1)$ that
\begin{align}
    \sum_{i=1}^{\mone} \frac{\Gamma_i p (1-p-q)^{\Gamma_i-1}}{1-(1-p)^{\Gamma_i}} = (1+o(1)) \sum_{i=1}^m \Gamma_i p (1-p-q)^{\Gamma_i-1} \label{eq_mone}
\end{align}
Hence,
\begin{align*}
m \Gamma_{\min}p(1-p-q)^{\Gamma_{\max}} \leq \Erw[W \mid \cG, \cE, |V_0^+|, \mone] \leq m \Gamma_{\max}p(1-p-q)^{\Gamma_{\min}-1}.
\end{align*}
Moreover, since $W$ is a binomial random variable, the Chernoff bound \pl{(e.g. \Lem~\ref{lemma_chernoff})} shows that
\begin{align*}
\pr\brk{\abs{W-\Erw[W \mid \cG, \cE, |V_0^+|]}>\sqrt{m}\log n}&\leq O(n^{-2}).
\end{align*}
Further, \Lem~\ref{Lemma_GammaMinMax} yields approximations for $\Gamma_{\min}$ and $\Gamma_{\max}$.
Now assume that $c<\log^{-2}2$. Using a similar reformulation as in \eqref{eq_mone}, we find that $p=(1+o(1)) k/n$. Thus, we have
\begin{align}
\Erw & [W| \cG, \cE, \cW] \nonumber \\
&= (1+o(1))m \frac{dn}{ck} \frac{k}{n} \exp \bc{(1+o(1)) \frac{dn}{ck} \log \bc{(1-k/n)\bc{1+n^{-\Omega(1)}}(1-(1-\exp(-d/c)))^{\Delta}}} \nonumber \\
&= (1+o(1))m \exp\bc{-d/c} \frac dc \left(1-(k/n)^{-d \log(1-\exp(-d/c))}\right)^{dn/(ck)} \label{W_second}
\end{align}

As \Lem~\ref{Lemma_Delta_too_small_big} shows, the optimal value of $d$ is a constant. For a fixed $c$ the same $d$ that maximizes $-d/c \log(1-\exp(-d/c))$ in \eqref{W_second}, also maximizes $\Erw [W| \cG, \cE, \abs{\zeroplus}]$. This maximum is attained at $d=c\log 2$. Consequently $p=o(q)$ and
\begin{align*}
q&\sim\bcfr kn^{c\log^22}.
\end{align*}

Hence,
\begin{align*}
\Erw[W \mid \cG, \cE, \cW]&\sim\frac{k\Delta}{2}\exp\bc{-(\log 2)\bcfr nk^{1-c\log^22}}=\exp(-n^{\Omega(1)}).
\end{align*}
As before, we find $\Erw[W] \to 0$ \whp\ \mhk{since $\Pr(\cW) = 1-o(1)$ and $\Pr(\cE) = \Omega(1/\sqrt{\Delta k})$} and Markov's inequality leads to $V_1^{--}=\emptyset$.
\Prop~\ref{Lemma_--} follows.
\color{black}
\end{proof}

\subsection{Proof of \Prop~\ref{Lemma_+}} \label{Proof_of_Lemma_+}

\mhk{By \Lem~\ref{Lemma_V_all}, we have $\abs{\zeroplus} \geq k \log n$ for $m<(1-\eps) \malg$. To prove \Prop~\ref{Lemma_+}, we need to show that for such $m$, we also have $\abs{V_0^{+,\Delta}} \geq k \log n$. We proceed in two steps. First, we show that every individual $x \in V$ is assigned to at least $\Delta - O(1)$ distinct tests. Second, we show that a constant fraction of individuals $x \in \zeroplus$ are assigned to exactly $\Delta$ tests establishing \Prop~\ref{Lemma_+}.}

\begin{proof}
Let $d^\star(x)$ be the number of distinct neighbors of a vertex $x$. We claim that \whp\ the following statements are true.
\begin{align*}
\min_{x\in V}d^\star(x)&\geq\Delta-2/\theta^2.
\end{align*}
The probability that a given $x\in V$ appears $\ell\ge2$ times in the same test is upper-bounded by
\begin{align*}
\binom\Delta\ell m^{1-\ell}&\leq\frac m{\ell!}\bcfr{d}{ck}^\ell=\frac{ck\log(n/k)}{\ell!}\bcfr{d}{ck}^\ell
\leq \frac{c (d/c)^\ell}{\ell!} n^{(1-\ell)\theta+o(1)}=o(1/n),
\end{align*}
provided that $\ell>1+1/\theta$.
\mhk{Moreover, the probability that $x$ appears in one test twice is upper-bounded by $\Delta \dot \Delta/m$.}
Thus, the probability that $x$ appears in at least $\ell$ tests at least twice is upper-bounded by
\color{black}
\begin{align*}
\sum_{i=\ell}^{\floor{\Delta/2}}\bc{\frac {\Delta^2} {m}}^i = (1+o(1)) \bc{\frac {\Delta^2} {m}}^\ell \leq (1+o(1)) \bcfr{O(\log^2 n)}{ck\log(n/k)}^\ell=n^{-\theta\ell+o(1)}=o(1/n),
\end{align*}
\color{black}
provided that $\ell>1/\theta$ \mhk{and since $m=ck\log(n/k)$ and $\Delta=d \log(n/k)$}.
The bound follows.

\mhk{By \Lem~\ref{Lemma_V_all}, we know that for $m<(1-\eps) \malg$, $\abs{\zeroplus} \geq k \log n$ \whp. Since the {\tt SCOMP} algorithm in its third stage selects the individual with the highest number of adjacent unexplained tests, we are left to show that also $\abs{{\zeroplus}^{,\Delta}} \geq k \log n$, which implies that \whp\ we erroneously classify a healthy individual as infected. The prior bounds} ensure that each individual is in at least $\Delta - O(1)$ tests. The question remains which fraction of individuals in $V_0^+$ are in ${\zeroplus}^{,\Delta}$.
In principle, it could be the case that most \mhk{potentially false positive} individuals of $V_0^+$ appear in less than $\Delta$ different tests. Indeed, it is more likely for such an individual in $V_0^+$ to be in fewer than $\Delta$ different tests since each additional test increases the probability for such an individual to be assigned to a negative test. However, we claim that a constant fraction of all \mhk{potentially false positive} individuals in $V_0^+$ \mhk{will have degree $\Delta$, thus} be in ${\zeroplus}^{,\Delta}$. To see this, let $p$ be the \mhk{maximum proportion of $\abs{{\zeroplus}^{,\Delta-i}}$ and $\abs{{\zeroplus}^{,\Delta-i+1}}$ for $i \in [2/\theta^2]$, i.e.}
\color{black}
\begin{align*}
    p = \max_{i \in [2/\theta^2]} \frac{\abs{{\zeroplus}^{,\Delta-i}}}{\abs{{\zeroplus}^{,\Delta-i+1}}}
\end{align*}
\color{black}
By conditioning on a test degree sequence $\Gamma_1, \dots, \Gamma_m$, we find
$$ p \geq (1-(1-(k/n))^{\Gamma_{\min}}) = \Theta(1), $$
as long as $c,d=\Theta(1)$, which by \Lem~\ref{Lemma_Delta_too_small_big} we can safely assume. Since each individual in $\zeroplus$ is in at least $\Delta-O(1)$ different tests and the probability of being in any number of different tests $\Delta, \Delta-1 \dots$ is constant, a constant fraction of individuals in $\zeroplus$ will be in exactly $\Delta$ tests. Since $\abs{\zeroplus} = \Omega(k \log n)$, the claim follows.
\end{proof}

\end{appendices}

\end{document}